\documentclass{article}
\usepackage[utf8]{inputenc}
\usepackage{amsmath}
\usepackage{amsthm}
\usepackage{amsfonts}
\usepackage{graphicx}
\usepackage[left=3.5cm,right=3.5cm]{geometry}
\usepackage{float} 
\usepackage{cite}
\usepackage{algorithm}
\usepackage{algorithmic}
\usepackage{color}
\usepackage{cleveref}
\usepackage{url}
\usepackage{lineno}
\usepackage{subfigure}
\usepackage{amsmath, amssymb, amsfonts}
\usepackage{abstract}
\usepackage{appendix}
\usepackage[numbers,sort&compress]{natbib}
\usepackage{dcolumn}
\usepackage{mathtools}
\usepackage{booktabs}
\usepackage{xcolor}

\allowdisplaybreaks
\usepackage[normalem]{ulem}

\title{Direct Inversion for Chi-Square Sampling, the Squared Bessel Process, and Applications}

\author{Simon J. A. Malham, Anke Wiese, and Yifan Xu\\[0.5cm]
  \small Maxwell Institute for Mathematical Sciences \\
  \small and School of Mathematical and Computer Sciences \\
  \small  Heriot-Watt University, Edinburgh EH14 4AS, UK}

\begin{document}

\setlength{\parindent}{0pt}
\setlength{\parskip}{1em}

\maketitle

\begin{abstract}
  In this paper we derive a new direct inversion method to simulate squared Bessel processes.
  Since the transition probability of these processes can be represented by a non-central chi-square distribution, 
  we construct an efficient and accurate algorithm to simulate non-central chi-square variables.
  In this method, the dimension of the squared Bessel process, equivalently the degrees of freedom of the chi-square distribution, is treated as a variable.
  We therefore use a two-dimensional Chebyshev expansion to approximate the inverse function of the central chi-square distribution with one variable being the degrees of freedom.
  The method is accurate and efficient for any value of degrees of freedom including the computationally challenging case of small values. 
  One advantage of the method is that noncentral chi-square samples can be generated for a whole range of values of degrees of freedom using the same Chebyshev coefficients.
  This property is especially useful in model calibration.
  The squared Bessel process is a building block for the well-known Cox-Ingersoll-Ross (CIR) processes,
  which can be generated from squared Bessel processes through time change and linear transformation.
  Our direct inversion method thus allows the efficient and accurate simulation of these processes, which are used as models in a wide variety of applications.   
\end{abstract}

\section{Introduction}
The squared Bessel process is a one-dimensional diffusion process that plays a fundamental role in modelling economic and financial
variables, see e.g., G{\"o}ing-Jaeschke and Yor \cite{going2003survey}. One important application is that it generates the well-known Cox-Ingersoll-Ross process,
a mean-reverting square-root process that serves as a model in a wide variety of applications (see Vere-Jones and Ozaki \cite{vere1982some} in seismology,
Hu, Kessler, Rappel and Levine \cite{hu2012input} in biology, and Cox, Ingersoll and Ross \cite{cox1985theory} in finance).
In this paper, we focus on its financial applications. The squared Bessel process also appears in describing key properties of Wiener processes, see for example Pitman and Yor 
(\cite{pitman1982decomposition} and \cite{pitman2003infinitely}), and  Pitman and Winkel \cite{pitman2018squared} for an overview.
The correspondence between the squared Bessel process and square-root processes, or Cox-Ingersoll-Ross processes, 
is given by linear scaling and time change (see Jeanblanc, Yor, and Chesney \cite{jeanblanc2009mathematical}, and Glasserman and Kim \cite{glasserman2011gamma}).
This simple correspondence means that results for squared Bessel process can be extended to square-root processes. 

The transition probability of the squared Bessel process 
can be represented by a scaled non-central chi-square distribution,
where the degree of freedom is given by what is known as the dimension of the squared Bessel process---the drift coefficient in the stochastic differential equation describing the process
(see details in Sec 6.2 by Jeanblanc, Yor, and Chesney \cite{jeanblanc2009mathematical}, and Revuz and Yor \cite{revuz2013continuous}).
To simulate the squared Bessel process, we thus start with its transition probability density.  

We propose a new direct inversion method to generate non-central chi-square samples, especially for the small degrees of freedom case, a computationally challenging situation 
that is relevant for financial applications, for example foreign exchange markets. Direct inversion was applied in Malham and Wiese \cite{r2}, who derived a new direct inversion method 
to sample Generalized Gaussian random variables and showed that sums of powers of these generate chi-square samples. 
In our new direct inversion method, we use the direct inversion method to simulate the central chi-square samples directly. 
Importantly, we treat the degrees of freedom as a variable.
Hence to simulate central chi-square variables, we establish a  two-dimensional Chebyshev approximation with one of the variables being the degrees of freedom of the chi-square distribution.
One of the advantages of the two-dimensional Chebyshev approximation is that a series of coefficients from the Chebyshev approximation is therefore valid for a whole range of values of the degrees of freedom.
This flexibility in our method means that efficient and accurate pricing of options with different parameter values can be achieved. 
Ga\(\ss\), Glau, Mahlstedt and Mair \cite{gass2018chebyshev} applied a multi-dimensional Chebyshev method directly to approximate option prices. 
Acceptance-rejection methods to generate non-central chi-square samples were developed 
by Ahrens and Dieter \cite{ahrens1974computer}, by Marsaglia and Tsang \cite{marsaglia2000simple}, and the generalized Marsaglia method by Malham and Wiese \cite{r2}.
The reason we choose the direct inversion method instead of an acceptance-rejection method is that to generate the chi-square samples with small degrees of freedom,
acceptance-rejection methods may reject a large number of samples, and this may be inefficient.

An important advantage of the two-dimensional Chebyshev approximation is that it provides analytical derivatives.
Since the degree of freedom \(\delta\) is one of the variables in the Chebyshev expansion, knowledge of this  derivative enables the sensitivity analysis of the squared Bessel process with respect to its dimension.
We then extend this analysis to the sensitivities of the mean-reverting square-root process with respect to its three parameters.
We apply these with our new direct inversion method to the calibration of the CIR process. 

Section 2 provides the background to squared Bessel processes; we describe the relationship between
the squared Bessel process and the Cox-Ingersoll-Ross process and the transition probability density of the squared Bessel process.
In Section 3 we introduce our new direct inversion method to generate non-central chi-square samples, and we include a comparison of the relative errors
in the first ten sample moments generated by
our new method to other methods. In Section 4, we derive the analytical sensitivities of the squared Bessel process and CIR process.
In Section 5 we focus on the application of squared Bessel processes to the pricing of options.
We illustrate the accuracy and efficiency of our new method in pricing path-independent put options on interest rates,
path-dependent Asian options and basket-options, and in the  robust calibration of the CIR process to European put, and Asian, options.
We conclude in Section 6 and provide an outlook. The appendix contains the Chebyshev coefficients for a range of values of degrees of freedom.

\section{The squared Bessel process and the CIR process}\label{sec2}
This section provides the background for the development 
of our direct inversion method and its application to the Cox-Ingersoll-Ross (CIR) process.
We present the relationship between the squared Bessel process and the CIR process, and we derive the transition probability of the squared Bessel process through the transition probability of the Cox-Ingersoll-Ross process.

The squared Bessel process is the solution of the stochastic differential equation, 
\begin{equation}\label{squared Bessel SDE}
\mathrm{d}Y_{t}=\delta\mathrm{d}t+2\sqrt{Y_{t}}\mathrm{d}B_{t},
\end{equation}
where $Y_0\ge 0$, and 
where \(\delta \ge 0\) is a constant, and \(B\) is a Wiener process. For every initial value $Y_0\ge 0$, a unique 
solution to the stochastic differential equation exists, see for example Jeanblanc, Yor and Chesney \cite{jeanblanc2009mathematical}. The parameter $\delta$ is known as the dimension of the squared Bessel process.
Assuming $Y_0>0$, when \(\delta \ge 2\), the squared Bessel process $Y$ never reaches zero,
and when \(0 < \delta < 2\), the squared Bessel process $Y$ reaches zero in finite time, and the zero boundary is instantaneously reflecting. 

 As introduced in Sec 6.3 by Jeanblanc, Yor, and Chesney \cite{jeanblanc2009mathematical}, we define the function $C=C(t)$ as follows: 
\begin{equation}\label{ct}
    C\coloneqq \frac{c^{2}}{4b}\left(1-e^{-bt}\right),
    \end{equation}
and set \[X_{t}\coloneqq e^{bt}Y_{C(t)},\] where \(a\) and \(c\) are positive constants satisfying \(\delta=4a/c^{2}\), and \(b\) is constant. 
Applying It{\^o}'s Lemma to \(X_{t}\), we have,
\begin{equation}\label{eq:CIR}
    \mathrm{d}X_{t}
=\left(a+bX_{t}\right)\mathrm{d}t+c\sqrt{X_{t}}\mathrm{d}W_{t},
\end{equation}
where $W$ is the Wiener process defined by, 
\begin{equation*}
  W_{t}\coloneqq 2\int_{0}^{t}(\sqrt{\mathrm{e}^{bs}}/c) \,\mathrm{d}B_{C(s)}.
\end{equation*}

We simulate the squared Bessel process by considering its transition probability density function following Malham and Wiese \cite{r2},
who simulate the CIR process by using its  transition probability density function. 
The transition probability of the squared Bessel process is shown in Theorem \ref{t1}, quoted from Chapter XI in Revuz and Yor \cite{revuz2013continuous},
who calculate this transition probability through the Laplace transform (for details see also Sec 6.2 in  Jeanblanc, Yor, and Chesney \cite{jeanblanc2009mathematical}). 

\newtheorem{Theorem}{Theorem}[section]
\newtheorem{Definition}{Definition}[section]
\begin{Definition}{\label{d1}}
The distribution function $F_{\chi_{\delta}^{2}(\lambda)} $ for 
a non-central chi-square random variable is given for \(y\in \mathbb{R}\) by,
\[
F_{\chi_{\delta}^{2}(\lambda)}(y)=\frac{e^{-\lambda/2}}{2^{\delta/2}}\sum_{i=0}^{\infty}\frac{(\lambda/2)^{i}}{i!2^{i}\Gamma(\delta/2+i)}\int_{0}^{y}x^{\delta/2+i-1}e^{-x/2}\mathrm{d}x,
\]
where \(\delta\) is the degrees of freedom, \(\lambda\) is the non-centrality parameter, and \(\Gamma(\cdot)\) is the standard gamma function.
\end{Definition}

\newtheorem{Proposition}{Proposition}[section]
\begin{Theorem}{\label{t1}} 
Let $Y$ be a squared Bessel process of dimension $\delta$.
Then for any \(u_{n+1}>u_{n}\), given \(Y_{u_{n}}\)  the conditional transition probability $P\left(Y_{u_{n+1}}<y|Y_{u_{n}}\right)$
is given as \((u_{n+1}-u_{n})\) times the non-central chi-square distribution with degrees of freedom \(\delta\) and non-centrality parameter 
\(\lambda:={Y_{u_n}}/(u_{n+1}-u_n)\), i.e., the conditional transition probability can be written as,
\[
P\left(Y_{u_{n+1}}<y|Y_{u_{n}}\right)=F_{\chi^{2}_{\delta}(\lambda)}\left(\frac{y}{u_{n+1}-u_{n}}\right).
\]
\end{Theorem}

To simulate the squared Bessel process from timestep \(u_{n}\) to \(u_{n+1}\),
we can generate a non-central chi-square sample with parameter \(\lambda=Y_{u_{n}}/(u_{n+1}-u_{n})\), then we have \(Y_{u_{n+1}}=(u_{n+1}-u_{n})\cdot\chi_{\delta}^{2}(\lambda)\).

\section{Chi-square sampling}
In order to design an efficient method to generate chi-square samples, we generate central chi-square samples first, and then use these 
to generate the non-central chi-square samples; see Malham and Wiese \cite{r2}. We propose a new direct inversion method to generate the central chi-square samples in this paper.

\subsection{Central chi-square sampling}

\begin{Definition}
We use $F_\delta$ to denote the central chi-square
distribution function with $\delta$ degrees of freedom.
\end{Definition}
In our direct inversion method, a tensored two-dimensional Chebyshev polynomial expansion is used to approximate the inverse central chi-square distribution function \(F^{-1}_{\delta}(u)\).
To approximate the inverse function accurately, we divide the range of 
\(F^{-1}_{\delta}(u)\) into several regions based on the behaviour of the shape of its cumulative function.
This is motivated by Moro \cite{fullmonte} and Malham and Wiese \cite{r2} who split the range of the inverse function of the generalized Gaussian distribution
into three regions based on the shape of the cumulative function to improve the approximation accuracy when using the direct inversion method to generate generalized Gaussian samples.
Malham and Wiese \cite{r2} use inflection points to determine how to divide the domain. 
However, for the central chi-square distribution
function  \(F_{\delta}(w)\), there are no positive real roots of
 \(F'''_{\delta}(w)=0\) and  \(F''''_{\delta}(w)=0\).
 Hence, we choose the two points \(w_{-}\) and \(w_{+}\) to divide the domain of the cumulative function \(F_{\delta}(w)\) into three regions by observing the behaviour of the chi-square cumulative function.
 In the first region, we have \(w\in[0, w_{-}]\) or \(u\in[0, F_{\delta}( w_{-})]\), where the cumulative function shows a sharp increasing trend with a large slope.
 In the middle region, we have \(w\in[w_{-}, w_{+}]\) or \(u\in[F_{\delta}( w_{-}), F_{\delta} (w_{+})]\), where the slope of the cumulative function decreases from a very large value to almost zero.
 In the tail region, we have \(w\in[w_{+}, \infty)\) or \(u\in[F_{\delta}(w_{+}),1)\), where the cumulative function is flat and nearly equals one. 
 To approximate the inverse central chi-square function \(F_{\delta}^{-1}(w)\) in each region, we use a tensored two-dimensional Chebyshev polynomial expansion.
 One potential issue is that since the parameter \(\delta\) is treated as one of the variables in the two-dimensional Chebyshev approximation, the boundary of each of the regions for the inverse variable \(u\) 
 depends on the variable \(\delta\). 
 This contrasts with the boundaries in a one-dimensional Chebyshev approximation. However, the regions for $u$ are uniquely determined for each value of $\delta$.  
 We also divide the range of the variable $\delta$, the degrees of freedom, in several intervals to achieve the accuracy as required.  
 In what  follows,  we denote by \([\Tilde{c},\Tilde{d}]\) such an interval for $\delta$,  where \(\Tilde{c}\) and \(\Tilde{d}\) are constants.
 We establish the two-dimensional Chebyshev approximation for each of the three regions of $u$.

\paragraph{Two-dimensional Chebyshev approximation.}
The inverse chi-square function \(F_{\delta}^{-1}(u)\) can be approximated as follows:
\begin{equation}\label{FU}
    F_{\delta}^{-1}(u)\approx\sum_{m=0}^{M}\sum_{n=0}^{N}c_{mn}T_{m}\left(\frac{2\delta-(\Tilde{d}+\Tilde{c})}{\Tilde{d}-\Tilde{c}}\right)T_{n}(k_{1}\xi(u)+k_{2}),
   \end{equation}
where \(T_{m}\) and $T_n$ are, respectively, the degree \(m\) and $n$ Chebyshev polynomials. 
Here $\xi$ is a scaling function, specified below, that depends 
in the region for $u$. The linear transformation 
$ 2\delta-(\Tilde{d}+\Tilde{c})/(\Tilde{d}-\Tilde{c})$ in $\delta$ and $k_1$ and 
$k_2$ are chosen to map the regions for $\delta$ and for $u$ 
to the interval $[-1, +1]$, on which the Chebyshev polynomials are defined.
The values of \(M\) and \(N\) depend on the value of \(\delta\).
A small value of \(\delta\) will lead to large values of \(M\) and \(N\).
That means when the Chebyshev approximation is for chi-square samples with a small degree of freedom, more coefficients are needed to achieve the prescribed approximation accuracy.
Motivated by Moro \cite{fullmonte} and Malham and Wiese \cite{malham2014efficient}, we define the scaled and shifted variable \(k_{1}\xi(u)+k_{2}\) for each region as follows.
In the first region, the inverse function is smooth and slowly varying, so we set \(\xi(u)=u\).
In the middle region, as the growth becomes steeper, we choose the logarithmic scaling variable to reduce curvature and improve approximation efficiency.
In the tail region, since the function grows almost singularly, we employ a $\log$-$\log$ scaling variable to further compress the domain and regularize
the behavior for accurate Chebyshev approximation.
The function \(\xi(u)\) is defined as,
\begin{equation}\label{xi}
    \xi(u)=
\begin{cases}
u, \hspace{5mm}& u\in[0, F_{\delta}(w_{-})),\\
\log((1-u)\Gamma(\delta/2)), & u\in[F_{\delta}( w_{-}), F_{\delta} (w_{+})),\\
\log(-\log((1-u)\Gamma(\delta/2))), & u\in [F_{\delta}(w_{+}),1-10^{-8}].\\
\end{cases}\end{equation}

We choose \(1-10^{-8}\) to be the upper boundary of the tail region.
That means we truncate the computational domain at \(u = 1 - \varepsilon\), where \(\varepsilon = 10^{-8}\).
This excludes a tail region with probability mass of order \(O(10^{-8})\).
However, when we prescribe the accuracy of the Chebyshev approximation with or within order \(O(10^{-8})\), the truncation does not introduce additional error beyond the target precision.
We label \(\xi^{(1)}(u)\), \(\xi^{(2)}(u)\), and \(\xi^{(3)}(u)\) for \(\xi(u)\) in the first, middle, and tail regions, 
and similarly for each of the three regions, we also label \(k_{j}^{(1)}\), \(k_{j}^{(2)}\), and \(k_{j}^{(3)}\) for each $k_j$, $j=1,\, 2$. 
The values of \(k_{1}^{(i)}\), and \(k_{2}^{(i)}\) for \(i=1, 2, 3\) are chosen such that:
\begin{enumerate}
    \item[(i)] \(k_{1}^{(1)}\xi^{(1)}(0)+k_{2}^{(1)}=-1\), and \(k_{1}^{(1)}\xi^{(1)}(F_{\delta}(w_{-}))+k_{2}^{(1)}=1\);
    \item[(ii)] \(k_{1}^{(2)}\xi^{(2)}(F_{\delta}(w_{-}))+k_{2}^{(2)}=-1\), and \(k_{1}^{(2)}\xi^{(2)}(F_{\delta}(w_{+}))+k_{2}^{(2)}=1\);
    \item[(iii)] \(k_{1}^{(3)}\xi^{(3)}(F_{\delta}(w_{+}))+k_{2}^{(3)}=-1\), and \(k_{1}^{(3)}\xi^{(3)}(F_{\delta}(1-10^{-8}))+k_{2}^{(3)}=1\).    
    \end{enumerate}
 
\paragraph{Coefficients of the two-dimensional Chebyshev expansion}
 As in Trefethen \cite{trefethen} and Scheiber \cite{chebyshev}, the coefficients of the two-dimensional Chebyshev expansion can be calculated by the following formula:
 \[
c_{mn}=
\begin{cases}
\frac{1}{\pi^{2}}I_{mn},\hspace{5mm}&m=0, n=0,\\
\frac{2}{\pi^{2}}I_{mn}, & m=1,n=0,\\
\frac{2}{\pi^{2}}I_{mn}, & m=0,n=1,\\
\frac{4}{\pi^{2}}I_{mn}, &m\geq1,n\geq1,
\end{cases}
\]
 with,
\[
\begin{split}   
I_{mn}=\int_{-1}^{1}\int_{-1}^{1}F^{-1}\left(\delta,u\right)\frac{T_{m}(\alpha)T_{n}(x)}{\sqrt{1-\alpha^{2}}\sqrt{1-x^{2}}}\mathrm{d}x\mathrm{d}\alpha,
\end{split}
\]
where the inverse chi-square function is written as a two-variable function \( F^{-1}(\delta,u)\coloneqq F_{\delta}^{-1}(u)\).
However, the inverse function \(F^{-1}(\delta,u)\) is unknown. Therefore, to calculate the coefficients, 
the inverse function \(F^{-1}(\delta, u)\) needs to be substituted when calculating the integral \(I_{mn}\).

Let \(x\coloneqq k_{1}\xi(u)+k_{2}\). Based on the function \(\xi(u)\) defined in equation (\ref{xi}), we have:
\begin{enumerate}
    \item[(i)] $u=\bigl(x-k_{2}^{(1)}\bigr)/k_{1}^{(1)}$ in the first region,
    \item[(ii)] $u=1-\exp\bigl(\bigl(x-k_{2}^{(2)}\bigr)/k_{1}^{(2)}\bigr)/\Gamma(\delta/2)$ in the middle region,
    \item[(iii)] $u=1-\exp\bigl(-\exp\bigl(\bigl(x-k_{2}^{(3)}\bigr)/k_{1}^{(3)}\bigr)\bigr)/\Gamma(\delta/2)$ in the tail region.
    \end{enumerate}

For \(\delta\in [\Tilde{c},\Tilde{d}]\), let \(\alpha\coloneqq\bigl(2\delta-(\Tilde{d}+\Tilde{c})\bigr)/(\Tilde{d}-\Tilde{c})\), so 
that \(\delta=\bigl((\Tilde{d}-\Tilde{c})\alpha+(\Tilde{d}+\Tilde{c})\bigr)/2\), and define $\eta^{(i)}(x, \delta)\coloneqq (x-k_2^{(i)})/k_1^{(i)}$ for $i=1, 2, 3$. 
Let \(I_{mn}^{(1)}\), \(I_{mn}^{(2)}\), and \(I_{mn}^{(3)}\) 
denote the integral \(I_{mn}\) of the coefficients \(c_{mn}\) in the three regions. We have the following expressions for the integral:
\[
\begin{split}
  I_{mn}^{(1)}&=\int_{-1}^{1}\int_{-1}^{1}F^{-1}\left(\frac{(\Tilde{d}-\Tilde{c})\cdot\alpha+(\Tilde{d}+\Tilde{c})}{2},\eta^{(1)}(x, \delta)\right)
  \cdot\frac{T_{m}(\alpha)T_{n}(x)}{\sqrt{1-\alpha^{2}}\sqrt{1-x^{2}}}\,\mathrm{d}x\mathrm{d}\alpha, \\
  I_{mn}^{(2)}&=\int_{-1}^{1}\int_{-1}^{1}F^{-1}\left(\frac{(\Tilde{d}-\Tilde{c})\alpha+(\Tilde{d}+\Tilde{c})}{2},1-\frac{\exp\left(\eta^{(2)}(x, \delta)\right)}{\Gamma(\delta/2)}\right)
  \cdot\frac{T_{m}(\alpha)T_{n}(x)}{\sqrt{1-\alpha^{2}}\sqrt{1-x^{2}}}\,\mathrm{d}x\mathrm{d}\alpha,\\
  I_{mn}^{(3)}&=\int_{-1}^{1}\int_{-1}^{1}F^{-1}\left(\frac{(\Tilde{d}-\Tilde{c})\alpha+(\Tilde{d}+\Tilde{c})}{2},1-\frac{1}{\Gamma(\delta/2)\exp\left(\exp\left(\eta^{(3)}(x, \delta)\right)\right)}\right)
  \cdot \frac{T_{m}(\alpha)T_{n}(x)}{\sqrt{1-\alpha^{2}}\sqrt{1-x^{2}}}\,\mathrm{d}x\mathrm{d}\alpha.
\end{split}
\]

Substituting \(u\) and \(\delta\) for \(x\) and \(\alpha\) yields,
\[
\begin{split}
 I_{mn}^{(i)}=\int_{c}^{d}G_{\delta}^{(i)}(u)\cdot\frac{T_{m}\left(\frac{2\delta-(\Tilde{d}+\Tilde{c})}{\Tilde{d}-\Tilde{c}}\right)}{\sqrt{\delta(\Tilde{d}+\Tilde{c})-\delta^{2}-dc}}\,\mathrm{d}\delta
\end{split}
\]
for \(i=1, 2,  3\). The function \(G_{\delta}^{(i)}(u)\) is defined as follows:
\[
\begin{split}
G_{\delta}^{(1)}(u)&\coloneqq\int_{0}^{F_{\delta}(w_{-})}F^{-1}(\delta,u)\cdot\frac{T_{n}\bigl(k_{1}^{(1)}\xi(u)+k_{2}^{(1)}\bigr)}{\sqrt{1-\bigl(k_{1}^{(1)}\xi(u)+k_{2}^{(1)}\bigr)^{2}}}\,\mathrm{d}u, \\
G_{\delta}^{(2)}(u)&\coloneqq\int_{F_{\delta}(w_{-})}^{F_{\delta}(w_{+})}F^{-1}(\delta,u)\cdot\frac{k_{1}^{(2)}\, T_{n}\bigl(k_{1}^{(2)}\xi(u)+k_{2}^{(2)}\bigr)}{(u-1)
  \,\xi(u)\,\sqrt{1-\bigl(k_{1}^{(2)}\xi(u)+k_{2}^{(2)}\bigr)^{2}}}\,\mathrm{d}u, \\
G_{\delta}^{(3)}(u)& \coloneqq\int_{F_{\delta}(w_{+})}^{1-10^{-8}}F^{-1}(\delta,u)\cdot\frac{k_{1}^{(3)}
  \, T_{n}\bigl(k_{1}^{(3)}\xi(u)+k_{2}^{(3)}\bigr)}{(1-u)\exp(\xi(u))\sqrt{1-\bigl(k_{1}^{(3)}\xi(u)+k_{2}^{(3)}\bigr)^{2}}}\,\mathrm{d}u,
 \end{split}
\]
where \(\xi(u)\) is defined in equation (\ref{xi}).

 Let \(w\coloneqq F^{-1}_{\delta}(u)\), so that 
\(F_\delta(w)=F_{\delta}\left(F^{-1}_{\delta}(u)\right)=u\), and 
 \(g_{\delta}^{(i)}(w)\coloneqq G_{\delta}^{(i)}(F_{\delta}(w))\). Then we have
\[
\begin{split}
  g_{\delta}^{(1)}(w)&=\int_{0}^{w_{-}}w\cdot
  \frac{ T_{n}\bigl(k_{1}^{(1)}\xi(F_{\delta}(w))+k_{2}^{(1)}\bigr)f_{\delta}(w)}{\sqrt{1-\bigl(k_{1}^{(1)}\xi(u)+k_{2}^{(1)}\bigr)^{2}}}\,\mathrm{d}w, \\
  g_{\delta}^{(2)}(w)&=\int_{w_{-}}^{w_{+}}w\cdot
  \frac{k_{1}^{(2)} \, T_{n}\bigl(k_{1}^{(2)}\xi(F_{\delta}(w))+k_{2}^{(2)}\bigr)f_{\delta}(w)}{(F_{\delta}(w)-1)\sqrt{1-\bigl(k_{1}^{(2)}\xi(F_{\delta}(w))+k_{2}^{(2)}\bigr)^{2}}}\,\mathrm{d}w,\\
  g_{\delta}^{(3)}(w)&=\int_{w_{+}}^{F_{\delta}^{-1}(1-10^{-8})}w\cdot
  \frac{k_{1}^{(3)} \, T_{n}\bigl(k_{1}^{(3)}\xi(F_{\delta}(w))+k_{2}^{(3)}\bigr)f_{\delta}(w)}{(1-F_{\delta}(w))\exp(\xi(F_{\delta}(w)))\sqrt{1-\bigl(k_{1}^{(3)}\xi(F_{\delta}(w))+k_{2}^{(3)}\bigr)^{2}}}\,\mathrm{d}w.
    \end{split}
\]
For the upper boundary \(F_{\delta}^{-1}(1-10^{-8})\) shown above, in practice we use a large value, here 20, 
such that \(F_{\delta}(20)>1-10^{-8}\), and use an indicator function to bound the inner function of the integral in the tail region as follows:
\[
\begin{split}
  g_{\delta}^{(3)}(w)=\int_{w_{+}}^{20}\mathbb{I}_{(F_{\delta}(w)\leq 1-10^{-8})}\, w \cdot
  \frac{k_{1}^{(3)} \, T_{m}\bigl(k_{1}^{(3)}\xi(F_{\delta}(w))+k_{2}^{(3)}\bigr)f_{\delta}(w)}{(1-F_{\delta}(w))\exp(\xi(F_{\delta}(w)))\sqrt{1-\bigl(k_{1}^{(3)}\xi(F_{\delta}(w))+k_{2}^{(3)}\bigr)^{2}}}\,\mathrm{d}w.
    \end{split}
\]

The integrals are now transformed to a representation that is suitable for  
the calculation of the coefficients of the two-dimensional Chebyshev expansion in the three regions.
These coefficients in the three regions are calculated using  Gauss-Konrod quadrature approximation using MATLAB with double precision.
The calculation accuracy of the double integration of the coefficients is of order $10^{-8}$. 

\theoremstyle{definition}
\newtheorem{Remark}{Remark}[section]

\begin{Remark}
The accuracy is influenced by the presence of singularities. 
When the chi-square variable \(w\) is close to the boundary 
of each region, or when the variable \(\delta\) is close to the upper and lower boundaries \(\Tilde{c}\) and \(\Tilde{d}\), 
the integrand tends to infinity, which decreases the accuracy in the calculation of the double integrals.
The influence of \(w\) is greater than that of \(\delta\) in the calculation of the double integral.
To improve the accuracy, one potential way is to increase our double precision (16 digits) arithmetic to a larger digits arithmetic (for example \(25\) digits).
Another way to obtain a more accurate estimate would be to consider an asymptotic expansion in a small interval that includes the singular endpoint.
\end{Remark}
\begin{Remark}
  The sufficiently large truncation value (20 in our implementation) is chosen so that the remaining tail probability is below the prescribed tolerance \(10^{-8}\).
  The contribution beyond this point is negligible and does not affect the overall numerical accuracy. Since the inverse distribution function 
  \(F_\delta^{-1}(u)\) depends on the parameter \(\delta\), the exact value of the boundary for the tail region varies with \(\delta\).
  Therefore, the truncation value must be chosen uniformly for all parameter values considered. The value 20 is selected such that 
\( F_{\delta}(20) \le 1-10^{-8}\)  for all \(\delta\) in the parameter range of interest. This ensures that the truncation error remains within the prescribed approximation tolerance.
\end{Remark}

 We have generated coefficients for different degrees of freedom by choosing different boundaries \(\Tilde{c}\) and \(\Tilde{d}\) for \(\delta\). 
 The truncation error of the Chebyshev expansion is bounded by the magnitude of the first neglected coefficient.
 The Chebyshev expansion is truncated once the absolute value of the Chebyshev coefficients falls below the prescribed tolerance level.
 In Table \ref{coefficient}, we show the number of coefficients required in the three different regions to achieve \(O(10^{-8})\) accuracy. 
Our direct inversion method achieves high accuracy approximation for any values of \(\delta\), especially for 
the challenging case of small values of \(\delta\). However, higher orders of Chebyshew polynomials are required for small 
\(\delta\). We also tested the region of \(\delta\in[0.2,1]\),
and observe that a smaller number of Chebyshev coefficients is required to achieve the same order of accuracy compared with the small values of \(\delta\) scenario.
In the Appendix, we list the series of coefficients in the first region, middle region, and tail region for \(\delta\in \left[0.1,0.2\right]\).
The CPU time to calculate the Chebyshev coefficients for the three regions ranges from 10 seconds to 50 seconds depending on the scaling of the shifted variables \(k_{1}\xi(u)+k_{2}\).
The computation time for the coefficients in the tail region is the longest among all three regions because of the most complex log-log scaling function \(\xi^{(3)}(u)\).

\begin{table}[H]
  \caption{We show the orders of two-dimensional Chebyshev polynomials required to achieve \(O(10^{-8})\) approximation accuracy.
    In each case, the pair shown is \((M, N)\), the orders of Chebyshev polynomials for the variables \(\delta\) and \(u\). The total number of Chebyshev coefficients in each region is \((M+1)\times (N+1)\).}
\centering
\begin{tabular}{|l|l|l|l|}

\hline
 \qquad \(\delta\)     & First region   & Middle region     & Tail region              \\ \hline
\([0.1, 0.2]\) & \((5, 13)\) & \((3,11)\) & \((3,13)\) \\ \hline
\([0.01,0.02]\) & \((4, 33)\)     & \((3,13)\)         & \((3,12)\)        \\ \hline
\([0.001, 0.002]\)      & \((4,39)\)       & \((2,13)\)        & \((3,13)\)                          \\ \hline
\end{tabular}

\label{coefficient} 
\end{table}

\paragraph{Two-dimensional Clenshaw's recurrence}
To generate central chi-square samples, we need to evaluate the sum of the two-dimensional Chebyshev series. We denote by,
\[
S(x,y)\coloneqq\sum_{m=0}^{M}\sum_{n=0}^{N}c_{m n}T_{m}(y)T_{n}(x),
\]
the Chebyshev polynomial in our approxmation for $F_\delta^{-1}$.
Recall \(x\) is the scaled variable of \(\delta\) with \(x=\bigl(2\delta-(\Tilde{d}-\Tilde{c})\bigr)/(\Tilde{d}+\Tilde{c})\),
and \(y\) is the scaled and shifted variable of \(u\) with different scaling as specified in equation \eqref{xi} in the first region, the middle region and the tail region.
A one-dimensional Chebyshev approximation can be evaluated efficiently using Clenshaw's recurrence formula (see Sec 5.8 in Press et al \cite{press2007numerical}).
As in Basu \cite{basu1973double}, to evaluate the two-dimensional Chebyshev series efficiently,
we use Clenshaw's recurrence for the variable \(y\) first and then again use the Clenshaw's recurrence formula for the second  variable \(x\).
The sum of all the terms of the Chebyshev series can be written as follows:
\[
 S(x,y)=\sum_{m=0}^{M}\sum_{n=0}^{N}c_{m n}T_{m}\left(y\right)T_{n}\left(x\right)=\sum_{n=0}^{N}a_{n}T_{n}(x),
\]
where $a_n=a_n(y)$ is given by,
\[
a_{n}\coloneqq \sum_{m=0}^{M}c_{mn}T_{m}(y).
\]
The Chebyshev polynomials \(T_{m}(y)\) and \(T_{n}(x)\) have the recurrence relation as follows:\[
T_{m}(y)-2yT_{m-1}(y)+T_{m-2}(y)=0,
\]
and \[
T_{n}(x)-2xT_{n-1}(x)+T_{n-2}(x)=0.
\]
Define the quantities \(d_{mn}\) with \(m=M,M-1, \cdots, 1\) and \(n=N,N-1, \cdots, 1\) by the following recurrence: 
\[
\begin{split}
&d_{M+2,n}=d_{M+1,n}=0,\\
  &  d_{mn}-2yd_{m+1,n}+d_{m+2,n}=c_{mn}.
    \end{split}
\]

By solving the recurrence formula above we obtain the expression of \(d_{1n}\) and \(d_{2n}\). Then we have,
 \[
a_{n}=yd_{1n}-d_{2n}+c_{0n}.
\]
Similarly, define the quantities \(g_{n}\) with \(n=N,N-1, \cdots, 1\) to solve the following recurrence:
\[
\begin{split}
   & g_{N+2}=g_{N+1}=0,\\
  &  g_{n}-2xg_{n+1}+g_{n+2}=a_{n} .  
    \end{split}
\]
Solving this recurrence to obtain \(g_{1}\) and \(g_{2}\), we then have the expression, 
\[
S(x,y)=xg_{1}-g_{2}+a_{0}.
\]

\subsection{Non-central chi-square sampling}\label{ssec:noncentral}
The method we use to generate the non-central chi-square samples using the central chi-square samples is taken from Malham and Wiese \cite{r2}. We describe this method briefly for completeness. 

We distinguish two cases for the non-centrality parameter $\lambda$: the case \(\lambda\leq10\) and the case \(\lambda> 10\).
When $\lambda \leq 10$,  then the non-central chi-square random variable \(\chi_{\delta}^{2}(\lambda)\) 
is decomposed as \(\chi_{\delta}^{2}(\lambda)\sim \chi_{\delta}^{2}+\chi_{2N}^{2}\) (see Siegel \cite{siegel1979noncentral}, and Johnson \cite{johnson1959extension}).
Here \(\chi_{\delta}^{2}\) can be generated by our direct inversion method proposed in the last section, and \(\chi_{2N}^{2}\) is a central chi-square random variable with degrees of freedom \(2N\),
where $N$ is a Poisson distributed random variable with mean \(\lambda/2\). We use the method in Glasserman \cite{glasserman2004monte} to generate \(\chi_{2N}^{2}\).
First we generate the Poisson distributed random variable \(N\), 
and then we generate \(N\) independent uniform random variables \(U_{1}, \dots, U_{N}\). 
Then  \(-2(\log(U_{1})+\dots+\log(U_{N}))\) has a \(\chi_{2N}^{2}\) distribution, and,
\[
\chi_{\delta}^{2}(\lambda) \sim -2\bigl(\log(U_{1})+\dots+\log(U_{N})\bigr)+Z,
\] 
where \(Z\) is the \(\chi_{\delta}^{2}\) random variable.

When \(\lambda>10\),  the non-central chi-square random variable \(\chi_{\delta}^{2}(\lambda)\) is decomposed as 
\(\chi_{\delta}^{2}(\lambda)\sim \chi_{\delta+2N}^{2}\sim \chi_{\delta+2\Bar{N}+2P}^{2}\),
using that \(N\) is a Poisson random variable, which can be expressed as a sum of two independent Poisson random variables \(\Bar{N}\) and \(P\) with mean \(10/2\) and mean \(\lambda/2-10/2\), respectively.
We generate a sample of the Poisson random variable \(\Bar{N}\), and consider the value of \(\Bar{N}\). 
If \(\Bar{N}\neq 0\), then we generate \(\Bar{N}-1\)
independent uniform random variables \(U_{1},\dots, U_{\Bar{N}-1}\) and two independent standard normal random variables \(V_{1}\) and \(V_{2}\).
Then by the additivity property  of non-central chi-square random variables, we have \(\chi_{\delta}^{2}(\lambda)\sim -2(\log(U_{1})+\dots+\log(U_{\Bar{N}-1}))+V_{1}^{2}+(V_{2}+\sqrt{\lambda-10})^{2}+Z\),
where \(Z\) is the \(\chi_{\delta}^{2}\) random variable. When \(\Bar{N}=0\), we set \(\bar{\lambda}=\lambda-10\).
If \(\bar{\lambda}>10\), we  repeat the process above, otherwise for \(\bar{\lambda}\leq 10\) we use the algorithm of generating the non-central chi-square random variable
with the non-centrality parameter smaller than \(10\) to generate the \(\chi_{\delta}^{2}(\bar{\lambda})\) random variable.

\subsection{Comparison}
We illustrate the accuracy of our direct inversion method by considering the relative errors in the first ten sample moments of the simulated non-central chi-square samples.
In Figure \ref{fig:momenterror} these relative errors 
are displayed for three values of the degrees of freedom \(\delta=0.1, 0.01,0.001\) (from top to bottom), and different values of the non-centrality parameter \(\lambda\).
The number of samples we use in each simulation is \(5\times 10^{7}\). 
As can be seen in the six panels, the direct inversion method  achieves conistently high accuracy among the different values of the degrees of freedom \(\delta\) and non-centrality parameter \(\lambda\).
The relative errors are dominated by the Monte Carlo error.
Comparing Figure \ref{fig:momenterror} with Figure 4 in Malham and Wiese \cite{r2},
where they show the relative errors in the first ten moments for the non-central chi-square samples simulated by three acceptance-rejection methods
(Ahrens-Dieter \cite{ahrens1974computer}, Marsaglia-Tsang \cite{marsaglia2000simple}, and Generalized  Marsaglia \cite{r2}),
the direct inversion method by Malham and Wiese \cite{r2}), and the method by Andersen \cite{andersen2008simple}, we observe that our direct inversion method has the same high order of accuracy.

We compare the CPU time of our direct inversion method with the acceptance-rejection method by Ahrens and Dieter \cite{ahrens1974computer} for generating \(10^{7}\) central chi-square samples.
As shown in Table \ref{CPU times}, our direct inversion method is consistently faster across different degrees of freedom.
Specifically, our direct inversion method is between \(35\%\) and \(54\%\) (rounded) faster than the Ahrens and Dieter method for values of delta of 0.1, 0,01, and  0.001.
Moreover, the CPU time of our method remains nearly constant for different values of \(\delta\), whereas the Ahrens–Dieter method shows noticeable variation.

\begin{table}[H]
  \caption{We compare our direct inversion method and the exact simulation method in Ahrens-Dieter \cite{ahrens1974computer} by generating central chi-square samples.
    The CPU times, in seconds, are the simulation of \(10^{7}\) central chi-square samples with different degrees of freedom.}
\centering
\begin{tabular}{|c|c|c|c|}
\hline\hline
Values of \(\delta\)     & 0.1   & 0.01     & 0.001              \\ \hline
Direct inversion method & \(0.961\) & \(0.998\) & \(0.992\) \\ \hline
Ahrens-Dieter's method & \(2.090\)     & \(1.782\)         & \(1.521\)                         \\ \hline\hline
\end{tabular}

\label{CPU times} 
\end{table}

\begin{figure}[H] 
\centering 
\includegraphics[width=1\textwidth]{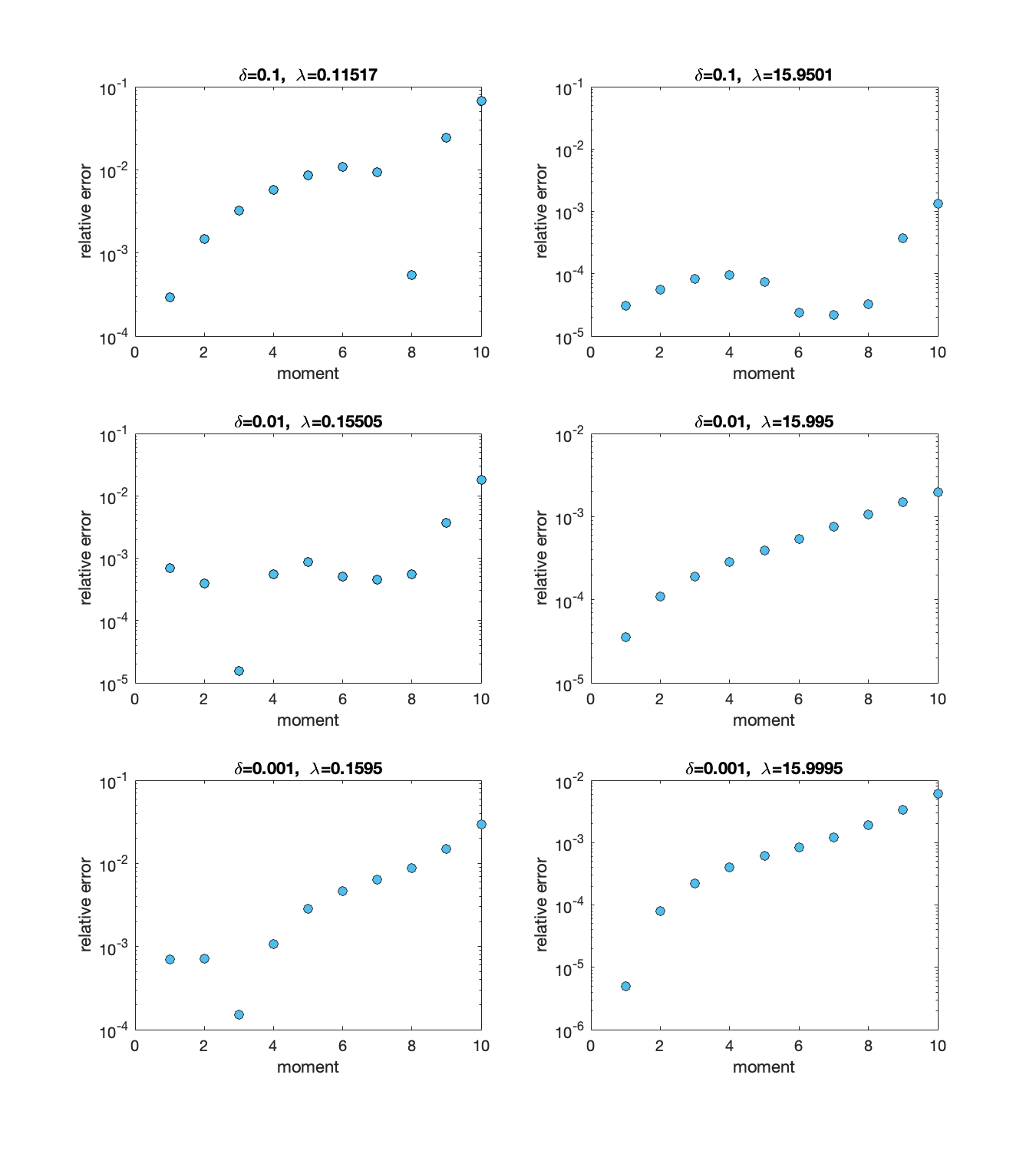} 
\caption{We show the relative errors in the first ten moments of the non-central chi-square samples using the direct inversion method.
  The simulations correspond to different values of degrees of freedom \(\delta\) and non-centrality parameter \(\lambda\).
  We test three values of degrees of freedom with \(\delta=0.1, 0.01, 0.001\), and the cases of
  \(\lambda\leq10\) and \(\lambda>10\) with values of \(\lambda=0.11517, 0.15505, 0.1595, 15.9501, 15.995, 15.9995\) (from top to bottom, and then from left to right).
  The number of samples in each case is \(5\times 10^{7}\).}
\label{fig:momenterror} 
\end{figure}

\section{Parameter sensitivities}\label{sec4.1}
Based on the new direct inversion method, we derive an analytical expression for the derivative of the squared Bessel process with respect to its dimension \(\delta\),
in other words the parameter sensitivity of the squared Bessel process. We then derive explicit expressions for the parameter sensitivities of the CIR model.

\subsection{Parameter sensitivity of squared Bessel processes}\label{sec4.11}
Recall that the squared Bessel process \(Y\) follows Equation (\ref{squared Bessel SDE}).
As introduced in Section \ref{sec2}, for the squared Bessel process at any time \(t\), we have,
\begin{equation*}
    Y_{t}=t\cdot \chi^{2}_{\delta}(\lambda),
\end{equation*}
where \(\chi^{2}_{\delta}(\lambda)\) is a non-central chi-square sample with the non-centrality parameter \(\lambda=Y_{0}/t\).
For the decomposition of the non-central chi-square variable \(\chi_{\delta}^{2}(\lambda)\sim \chi_{\delta}^{2}+\chi_{2N}^{2}\),
only the first central chi-square variable \(\chi^{2}_{\delta}\) depends on the parameter \(\delta\).
The central chi-square variable \(\chi^{2}_{\delta}\) is generated by the central chi-square inverse function \(F^{-1}_{\delta}(U)\),
  which is defined in Equation (\ref{FU}), where \(U\sim \text{Unif }[0,1]\).
Then we define the derivative of the squared Bessel process with respect to the degree of freedom parameter as follows:
\begin{equation*}
    \partial_\delta Y_{t}\coloneqq t\cdot\partial_\delta F^{-1}_{\delta},
\end{equation*}
where $\partial_\delta\coloneqq\partial/\partial\delta$.
The central chi-square inverse function \(F^{-1}_{\delta}\) is approximated by a two-dimensional Chebyshev expansion
with the degree of freedom parameter \(\delta\) as one of the variables, see Equation (\ref{FU}).
Hence the first derivative of the central chi-square variable with respect to the parameter \(\delta\) can be approximated as follows:
  \begin{align}
      \partial_\delta\biggl(\sum_{m=0}^{M}\sum_{n=0}^{N}c_{mn}T_{m}&\Biggl(\frac{2\delta-(\Tilde{d}+\Tilde{c})}{\Tilde{d}-\Tilde{c}}\biggr) T_{n}\left(k_{1}\xi(u)+k_{2}\right)\Biggr)\nonumber\\
      \qquad\qquad=&\;\sum_{m=0}^{M}\sum_{n=0}^{N}c_{mn}\partial_\delta T_{m}\biggl(\frac{2\delta-(\Tilde{d}+\Tilde{c})}{\Tilde{d}-\Tilde{c}}\biggr) \cdot T_{n}\left(k_{1}\xi(u)+k_{2}\right)\nonumber\\
      \qquad\qquad&\;+\sum_{m=0}^{M}\sum_{n=0}^{N}c_{mn}T_{m}\biggl(\frac{2\delta-(\Tilde{d}+\Tilde{c})}{\Tilde{d}-\Tilde{c}}\biggr)\cdot\partial_\delta T_{n}\left(k_{1}\xi(u)+k_{2}\right).
    \label{approx_2d}
\end{align}
We write the scaling functions inside of the Chebyshev polynomials $T_{m}$ and $T_{n}$ as functions of the target parameter $\delta$ by setting,
\begin{equation*}
    x(\delta)\coloneqq\frac{2\delta-(\Tilde{d}+\Tilde{c})}{\Tilde{d}-\Tilde{c}}
    \qquad\text{and}\qquad
    y(\delta)\coloneqq k_{1}\xi(u)+k_{2}.
\end{equation*}
For \(n=0,1,\cdots, N\), we set for $x=x(\delta)$,
\begin{equation*}
\begin{split}
    f_{1}(x)\coloneqq\sum_{m=0}^{M}c_{mn} T_{m}(x).
    \end{split}
\end{equation*}
The derivative of the Chebyshev expansion \(f_{1}(x(\delta))\) can be represented by a new series of coefficients generated
from the original Chebyshev coefficients as follows for $x=x(\delta)$ 
(see details in Section 5.9 by Press et al. \cite{press2007numerical}):
\begin{equation*}
    \partial_x f_{1}(x)=\sum_{m=0}^{M}c_{mn}\partial_x T_{m}(x)=\sum_{m=0}^{M-1}b_{mn} T_{m}(x),
\end{equation*}
with 
\begin{equation*}
b_{mn}=b_{m+2,n}+2(m+1)c_{m+1,n},
\end{equation*}
where \(n=0,1,\cdots ,N\), and \(m=0,1,\cdots ,M-1\), and \(b_{M+1,n}=b_{M,n}=0\).

Similarly, for \(m=0,1,\cdots,M\), we set for $y=y(\delta)$,
\begin{equation*}
f_{2}(y)\coloneqq\sum_{n=0}^{N}c_{mn}T_{n}(y).
\end{equation*}
The expression of the derivative of the Chebyshev expansion \(f_{2}(y)\) for $y=y(\delta)$ is,
\begin{equation*}
  \partial_y f_{2}(y)
  =\sum_{n=0}^{N}c_{mn}\partial_y T_{n}(y)=\sum_{n=0}^{N-1}\Tilde{b}_{mn} T_{n}(y),
\end{equation*}
with new coefficients,
\begin{equation*}
\Tilde{b}_{m,n}=\Tilde{b}_{m,n+2}+2(n+1)c_{m,n+1},
\end{equation*}
where \(n=0,1,\cdots ,N-1\), \(m=0,1,\cdots ,M\), and \(\Tilde{b}_{m,N+1}=\Tilde{b}_{m,N}=0\).

Thus, with the new series of coefficients \(b_{mn}\) and \(\Tilde{b}_{mn}\),
the terms on the right in equation (\ref{approx_2d}) have the following expression for $x=x(\delta)$ and $y=y(\delta)$,
\begin{equation}\label{dcheby}
    \begin{split}
      &\sum_{n=0}^{N}\partial_x f_{1}(x)\cdot\partial_\delta x\cdot T_{n}(y)+\sum_{m=0}^{M}T_{m}(x)\cdot\partial_y f_{2}(y)\cdot\partial_\delta y\\[8pt]=&
      \left(\sum_{n=0}^{N}\left(\sum_{m=0}^{M-1}b_{mn} T_{m}\left(x(\delta)\right)\right)\cdot T_{n}\left(y(\delta)\right)\right)\,\partial_\delta x 
      +\left(\sum_{m=0}^{M}\left(\sum_{n=0}^{N-1}\Tilde{b}_{mn}\cdot T_{n}\left(y(\delta)\right)\right)T_{m}\left(x(\delta)\right)\right)\,\partial_\delta y, 
    \end{split}
\end{equation}
which is the approximate derivative of the central chi-square inverse function with respect to the parameter \(\delta\).
The derivatives of the scaling functions are calculated as follows: 
\begin{equation*}
\partial_\delta x=\frac{2}{d-c},
\end{equation*}
and
\begin{equation*}
  \partial_\delta y=\partial_\delta\bigl(k_{1}\xi(u)+k_{2}\bigr)
  =\xi(u)\cdot\partial_\delta k_{1}+k_{1}\cdot\partial_\delta\xi(u)+\partial_\delta k_{2}. 
\end{equation*} 

We recall from \eqref{xi} and the equations defining $k_1$ and $k_2$ immediately thereafter, that in general, $\xi=\xi(u)$, $k_{1}$ and $k_{2}$ are $\delta$-dependent.
Hence hereafter, we explicitly denote $\xi=\xi(u,\delta)$, $k_{1}=k_{1}(\delta)$ and $k_{2}=k_{2}(\delta)$. 
The derivative of $\xi=\xi(u,\delta)$ with respect to $\delta$ is given by,
\begin{equation*}
\partial_\delta \xi=
\begin{cases}
0, \hspace{5mm}& u\in[0, F_{\delta}(w_{-})),\\[8pt]
\tfrac12\varphi^{(0)}(\delta/2), & u\in[F_{\delta}( w_{-}), F_{\delta} (w_{+})),\\[8pt]
\varphi^{(0)}(\delta/2)/2\log((1-u)\Gamma(\delta/2)), & u\in [F_{\delta}(w_{+}),1-10^{-8}],\\
\end{cases}
\end{equation*}
where $\varphi^{(0)}$ is the polygamma function of order $0$.

Let \(u_{1}\) and \(u_{2}\) denote the lower and upper boundary points of the variable \(u\) among the three regions.
The functions $k_1=k_{1}(\delta)$ and $k_2=k_{2}(\delta)$ are chosen such that,
\begin{equation*}
  k_{1}(\delta)\cdot\xi(u_{1},\delta)+k_{2}(\delta)=-1,
  \qquad \text{and}\qquad
  k_{1}(\delta)\cdot\xi(u_{2},\delta)+k_{2}(\delta)=1.
\end{equation*}
The derivatives of $k_1=k_{1}(\delta)$ and $k_2=k_{2}(\delta)$ with respect to $\delta$ in each of the three regions are as follows:
\begin{equation*}
  \partial_\delta k_{1}(\delta)=\frac{2}{\bigl(\xi(u_{1},\delta)-\xi(u_{2},\delta)\bigr)^{2}}
  \cdot\bigl(\partial_\delta\xi(u_{1},\delta)-\partial_\delta\xi(u_{2},\delta)\bigr),
\end{equation*}
and 
\begin{equation*}
  \partial_\delta k_{2}(\delta)=-\partial_\delta k_{1}(\delta)\cdot\xi(u_{2},\delta)-\partial_\delta\xi(u_{2},\delta)\cdot k_{1}(\delta).
\end{equation*}
 Let \(u_{j}\), for \(j=1,2\), be the lower or upper boundary. We have,
\begin{equation}\label{xii}
\begin{split}
  \partial_\delta \xi(u_{j},\delta)
  =&\;\partial_\delta\xi(u_{j},\delta)+\partial_{u_j}\xi(u_{j},\delta)\cdot\partial_\delta u_{j}\\[8pt]
  =&\;\begin{cases}
\partial_\delta u_{j}, \hspace{5mm}& u\in[0, F_{\delta}(w_{-})),\\[8pt]
\tfrac12\varphi^{(0)}(\delta/2)-\frac{1}{1-u_{j}}\cdot\partial_\delta u_{j}, & u\in[F_{\delta}( w_{-}), F_{\delta} (w_{+})),\\[8pt]
  \frac{\varphi^{(0)}(\delta/2)}{2\log((1-u_{j})\Gamma(\delta/2))}+\frac{(1-u_{j})\cdot (\Gamma(\delta/2))^{2}}{\log((1-u_{j})\Gamma(\delta/2))}\cdot\partial_\delta u_{j},
  & u\in [F_{\delta}(w_{+}),1-10^{-8}].
\end{cases}
\end{split}
\end{equation}
For the central chi-square cumulative function $F_{\delta}$, we have:
\begin{enumerate}
  \item \(\partial_\delta u_{1}=0\) and \(\partial_\delta u_{2}=\partial_\delta F_{\delta}(w_{-})\) in the first region;
  \item \(\partial_\delta u_{1}=\partial_\delta F_{\delta}(w_{-})\) and \(\partial_\delta u_{2}=\partial_\delta F_{\delta}(w_{+})\) in the middle region;
  \item \(\partial_\delta u_{1}=\partial_\delta F_{\delta}(w_{+})\) and \(\partial_\delta u_{2}=0\) in the tail region.
\end{enumerate}\medskip

\begin{Remark}
For the derivative of the Chebyshev expansion in Equation (\ref{dcheby}), if we set for $x=x(\delta)$ and $y=y(\delta)$,
\begin{equation*}
b^{*}_{mn}\coloneqq b_{mn}\cdot\partial_\delta x+\Tilde{b}_{mn}\partial_\delta y,
\end{equation*}
then \(b^{*}_{mn}\)  can be treated as a new coefficient for a two-dimensional Chebyshev expansion.
The approximation of \(\partial_\delta F^{-1}_{\delta}(u)\) can be re-written in the Chebyshev formula format:
\begin{equation*}
\begin{split}
  &\sum_{n=0}^{N-1}\sum_{m=0}^{M-1}b^{*}_{mn}T_{m}(x(\delta))\cdot T_{n}(y(\delta)) +\left(\sum_{m=0}^{M-1}b_{mN}T_{m}(x(\delta))\right)
  \cdot T_{N}(y(\delta))\cdot\partial_\delta x(\delta)\\
        &\qquad\qquad\qquad+\left(\sum_{n=0}^{N-1}\Tilde{b}_{Mn}\cdot T_{n}(y(\delta))\right)\cdot T_{M}(x(\delta))\cdot\partial_\delta y(\delta).
    \end{split}
\end{equation*}
\end{Remark}

\subsection{Parameter sensitivities of the CIR model}\label{section4.2}
In this section, we derive the analytical sensitivities of the CIR model, see Equation (\ref{eq:CIR}), with respect to its three parameters.
  For a CIR process, parameter \(a\) is the mean-reversion level, \(b\) is the negative of the rate of convergence, and \(c\) is the volatility parameter.

\paragraph{Sensitivity with respect to the mean-reversion level.}
We consider the sensitivity of the CIR process with respect to the mean-reversion level parameter \(a\) via the degrees of freedom parameter \(\delta\) in the squared Bessel process.
Since the CIR process is simulated through the squared Bessel process and we have the parameter relation \(\delta=4a/c^{2}\),
the derivative of the CIR process with respect to the parameter \(a\)  can be deduced via the derivative with respect to the parameter \(\delta\) by the chain rule as follows:
\begin{equation*}
    \partial_a X_{t}=\partial_\delta X_{t}\cdot\partial_a\delta =(4/c^{2})\cdot\partial_\delta X_{t}.
\end{equation*}
We interpret these derivatives of the CIR process in the same way as in Section \ref{sec4.11} for the squared Bessel process.
 For a squared Bessel process \(Y\) with the degree of freedom parameter \(\delta\),  since we have $X_{t}=\mathrm{e}^{bt}\cdot Y_{C(t)}$, the derivative can be expressed in the following form:
\begin{equation*}
  \partial_\delta X_{t}=\mathrm{e}^{bt}\cdot\partial_\delta Y_{C(t)}.
  \end{equation*}

For the time changed squared Bessel process \(Y_{C(t)}\),  the time change function \(C(t)\) is independent of the parameter \(\delta\).
Based on the analysis of the derivative of the squared Bessel process with respect to the parameter \(\delta\) in Section \ref{sec4.11}, we have the following:
\begin{equation*}
\partial_\delta Y_{C(t)}=C(t)\cdot\partial_\delta F^{-1}_{\delta},
\end{equation*}
where \(F^{-1}_{\delta}\) is the inverse function of the central chi-square distribution with degrees of freedom \(\delta\). 
The detailed expression of \(\partial_\delta F^{-1}_{\delta}\) is given in Section \ref{sec4.11}.
Thus, we have established an expression of the derivative of the CIR process with respect to the parameter \(a\).

\paragraph{Sensitivity with respect to the rate of convergence.} For a CIR process given by Equation (\ref{eq:CIR}), the rate of convergence is \(-b\).
The parameter \(b\) occurs in both the time change function \(C(t)\) and the linear transformation from the squared Bessel process \(Y\) to the CIR process \(X\).
Due to the complex positions of the parameter \(b\), we cannot derive a direct expression for the derivative of the CIR process with respect to the parameter \(b\). 
Here, we use the score function to estimate the sensitivity with respect to the parameter \(b\)
following the likelihood ratio framework described in Glasserman \cite[Section 7.3]{glasserman2004monte}.
We illustrate this approach using examples of a European put option and an Asian option.

For a CIR process \(X\), the price of a European put option with asset price \(X_{T}\) at expiration time \(T\) can be calculated using the following formula:
\begin{equation*}
\mathrm{E}\left[(K-X_{T})^{+}\right]=\int_{0}^{K}(K-x)^{+}\cdot g_{b}(x)\mathrm{d}x,
\end{equation*}
where \(K\) is the strike price, and \(g_{b}(x)\) is the density function of $X_T$.
The derivative of the option price with respect to the parameter \(b\) is given by
\begin{equation}\label{scoref}
    \begin{split}
      \partial_b\mathrm{E}\left[(K-X_{T})^{+}\right]
      =\int_{0}^{K}(K-x)^{+}\cdot\partial_b g_{b}(x)\,\mathrm{d}x
      =\mathrm{E}\biggl[(K-X_{T})^{+}\cdot\frac{\partial_b g_{b}(X_{T})}{g_{b}(X_{T})}\biggr].
    \end{split}
\end{equation}
To evaluate the expectation in Equation (\ref{scoref}), we need to know the expression of \(\partial_b g_{b}(X_{T})/g_{b}(X_{T})\).

As introduced in Section \ref{sec2}, the simulation of the CIR process is closely related to the non-central chi-squared samples.
Recall the density function of the non-central chi-square distribution with the degrees of freedom \(\delta\)
and the non-centrality parameter \(\lambda\) in Definition \ref{d1} is given as follows:
\begin{equation*}
f_{\chi^{2}_{\delta}(\lambda)}\left(z\right)=\tfrac12\cdot e^{-(z+\lambda)/2}\cdot(z/\lambda)^{\delta/2-1/2}\cdot I_{\delta/2-1}\bigl(\sqrt{\lambda z}\bigr),
\end{equation*}
where the modified Bessel function can be written as,
\begin{equation*}
I_{\delta}(x)=\sum_{j=0}^{\infty}\frac{(x/2)^{2j+\delta}}{j!\,\Gamma(\delta+j+1)}.
\end{equation*}

The CIR process can be simulated using the method proposed by Malham and Wiese \cite{r2},
who simulate the CIR process based on its known transition probability function.
In our scenario of the European put option, we only consider the time points at \(t=0\) and at \(t=T\).
Recall the transition probability is given as follows, see for example Proposition 4.1 in Malham and Wiese \cite{r2}. 
\begin{Proposition} {\label{cirTi}} 
  For a CIR process \(X\), the conditional distribution of \(X_{T}\) given the initial value \(X_{0}\) is \(c^{2}(\mathrm{e}^{bT}-1)/4b\)
  times a non-central chi-square distribution with degrees of freedom \(\delta\) and a non-centrality parameter \(\lambda\),
  where,
  \begin{equation}\label{eq:lambdadef}
    \lambda=\frac{4bX_{0}\mathrm{e}^{bT}}{c^{2}(\mathrm{e}^{bT}-1)}.
  \end{equation}
\end{Proposition}

Thus, $X_{T}\cdot4b/\bigl(c^{2}(\mathrm{e}^{bT}-1)\bigr)$ is a $\chi^{2}_{\delta}(\lambda)$ distributed random variable,
with the non-centrality parameter $\lambda$ given by \eqref{eq:lambdadef}.
Thus the density function of the CIR process is given by,
\begin{equation*}
    g_{b}(x)\coloneqq f_{\chi^{2}_{\delta}(\lambda)}\left(\frac{4bx}{c^{2}(\mathrm{e}^{bT}-1)}\right)\cdot \frac{4b}{c^{2}(\mathrm{e}^{bT}-1)}.
\end{equation*}
To simplify, we set:
\begin{equation*}
    z\coloneqq \frac{4bx}{c^{2}(\mathrm{e}^{bT}-1)},
\end{equation*}
and define the function $h=h(b)$ as,
\begin{equation*}
    h\coloneqq \frac{4b}{c^{2}(\mathrm{e}^{bT}-1)}.
\end{equation*}
Then the density function of the CIR process is,
\begin{equation*}
    g_{b}(x)=f_{\chi^{2}_{\delta}(\lambda)}(z)\cdot h(b).
    \end{equation*}
By the chain rule, we have,
\begin{equation*}
    \begin{split}
   \frac{\partial_b g_{b}(x)}{g_{b}(x)} =\frac{\partial_b f_{\chi^{2}_{\delta}(\lambda)}(z)}{f_{\chi^{2}_{\delta}(\lambda)}(z)}+\frac{\partial_b h(b)}{h(b)}.
   \end{split}
\end{equation*}
Recall that $\lambda$ is given by the formula in \eqref{eq:lambdadef}.
Hence the density function of the non-central chi-square distribution has the following form:
\begin{equation*}
\begin{split}
    f_{\chi^{2}_{\delta}(\lambda)}(z)=f_{1}(b)\cdot I_{\delta/2-1}\bigl(X_{0}\mathrm{e}^{bT}/x\bigr)
    \end{split}
\end{equation*}
with 
\begin{equation*}
    f_{1}(b)=\frac{1}{2}\cdot e^{\frac{2b(x+X_{0}\mathrm{e}^{bT})}{c^{2}(1-\mathrm{e}^{bT})}}\cdot \left(\frac{x}{X_{0}\mathrm{e}^{bT}}\right)^{\frac{\delta-1}{2}}.
\end{equation*}
Again, by the chain rule, the derivative with respect to the parameter \(b\) is,
\[\partial_b f_{\chi^{2}_{\delta}(\lambda)}(z)=
\partial_b f_{1}(b)\cdot I_{\delta/2-1}\bigl(X_{0}\mathrm{e}^{bT}/x\bigr)+f_{1}(b)\cdot\partial_bI_{\delta/2-1}\bigl(X_{0}\mathrm{e}^{bT}/x\bigr).
\]
Then we have the following:
\begin{equation*}
\frac{\partial_b f_{\chi^{2}_{\delta}(\lambda)}(z)}{f_{\chi^{2}_{\delta}(\lambda)}(z)}=
\frac{\partial_b f_{1}(b)}{f_{1}(b)}+\frac{\partial_bI_{\delta/2-1}\bigl(X_{0}\mathrm{e}^{bT}/x\bigr)}{I_{\delta/2-1}\bigl(X_{0}\mathrm{e}^{bT}/x\bigr)},
\end{equation*}
with,
\begin{equation}
 \frac{\partial_b f_{1}(b)}{f_{1}(b)}
  =\tfrac12 T(1-\delta)
  +\frac{2(x+X_{0}\mathrm{e}^{bT})+2bX_{0}T\mathrm{e}^{bT}}{c^{2}(1-\mathrm{e}^{bT})}
  +\frac{2(x+X_{0}\mathrm{e}^{bT})bT\mathrm{e}^{bT}}{c^{2}(1-\mathrm{e}^{bT})^{2}},\label{eq:f1term}
\end{equation}
and 
\begin{equation}
  \frac{\partial_bI_{\delta/2-1}\bigl(X_{0}\mathrm{e}^{bT}/x\bigr)}{I_{\delta/2-1}\bigl(X_{0}\mathrm{e}^{bT}/x\bigr)}
  =\frac{I_{\delta/2}\bigl(X_{0}\mathrm{e}^{bT}/x\bigr)}{I_{\delta/2-1}\bigl(X_{0}\mathrm{e}^{bT}/x\bigr)}
  \cdot \frac{X_{0}T\mathrm{e}^{bT}}{x}+\delta/2-1. \label{eq:besselterm}
    \end{equation}
Here, equation (\ref{eq:besselterm}) follows from the formula of the derivative of the modified Bessel function,
\begin{equation*}
\partial_x I_{\delta}(x)=\frac{\delta}{x}I_{\delta}(x)+I_{\delta+1}(x).   
\end{equation*}
Finally, we calculate the derivative of $h=h(b)$ as follows,
\begin{equation}
    \frac{\partial_b h(b)}{h(b)}=\frac{1}{b}-\frac{T\mathrm{e}^{bT}}{\mathrm{e}^{bT}-1}.\label{eq:hterm}
\end{equation}
In summary, we observe that $\partial_b g_{b}(x)/g_{b}(x)$ is given by the sum of the terms on the righthand side in \eqref{eq:f1term}, \eqref{eq:besselterm} and \eqref{eq:hterm}.
We can thus calculate the derivative in \eqref{scoref} via the righthand sides in \eqref{eq:f1term}, \eqref{eq:besselterm} and \eqref{eq:hterm} by simulating the asset price $X_{T}$.

Next, we apply the score function to an Asian option price to estimate the sensitivity of the parameter \(b\).
We assume that the payoff function of the Asian option is,
\begin{equation*}
f(X_{t_{1}},\ldots, X_{t_{M}})\coloneqq\left(K-\frac{1}{M}\sum_{m=1}^{M}X_{t_{m}}\right)^{+},
\end{equation*}
where \(X_{t_{m}}\) is the asset price at the monitoring times  \(0=t_{0}<t_{1}<\ldots<t_{M}=T\), and \(M\) is the number of monitoring times, and \(K\) is the strike price.
The derivative of the option price with respect to the parameter \(b\) can be written as follows:
\begin{equation*}
\begin{split}
    \partial_b E\bigl[f(X_{t_{1}},\ldots, X_{t_{M}})\bigr]=E\bigl[f(X_{t_{1}},\ldots, X_{t_{M}})\cdot\partial_b\log(g(X_{t_{1}},\ldots,X_{t_{M}}))\bigr],
   \end{split}
\end{equation*}
where \(g(X_{t_{1}},\ldots,X_{t_{M}})\) is the joint density function of the asset prices \(X_{t_{0}},X_{t_{1}},\ldots,X_{t_{M}}\).
We have,
\begin{equation*}
g(x_{t_{1}},\ldots,x_{t_{M}})=g(x_{1}|x_{0})\cdot g(x_{2}|x_{1})\cdot \ldots \cdot g(x_{M}|x_{M-1}),
\end{equation*}
and thus, 
\begin{equation*}
  \partial_b\log(g(x_{t_{1}},\cdots,x_{t_{M}}))
  =\partial_b\log(g(x_{t_{1}}|x_{t_{0}}))+\partial_b\log(g(x_{t_{2}}|x_{t_{1}}))+\cdots+\partial_b\log(g(x_{t_{M}}|x_{t_{M-1}})).
\end{equation*}
Each of the partial derivatives, 
 \begin{equation*}
 \partial_b\log(g(x_{t_{m}}|x_{t_{m-1}}))=\frac{\partial_b g(x_{t_{m}}|x_{t_{m-1}})}{g(x_{t_{m}}|x_{t_{m-1}})},
\end{equation*} 
 can be calculated using the same approach for the European put option case above via the righthand sides in \eqref{eq:f1term}, \eqref{eq:besselterm} and \eqref{eq:hterm},
 by substituting \(X_{0}\) with \(x_{t_{m-1}}\), and substituting the asset price variable \(x\) with \(x_{t_{m}}\).

\paragraph{Sensitivity with respect to the volatility parameter.}
To consider the first derivative of the CIR process with respect to the volatility  parameter \(c\),  let \(Z_t\coloneqq \kappa X_t\) for a constant \(\kappa\).
 The process $Z$ is also a CIR process with SDE,
\begin{equation}
\mathrm{d}Z_{t}=(\kappa a+bZ_{t})\mathrm{d}t+c\sqrt{\kappa}\sqrt{Z_{t}}\mathrm{d}W_{t}.
\label{newcir}
\end{equation}
We let \(\hat{a}\), \(\hat{b}\), and \(\hat{c}\) denote the parameters of the new CIR process by setting,
 \begin{equation*}
     \hat{a}=\kappa a, \quad\quad \hat{b}=b, \quad \quad \text{and} \; \hat{c}=c\sqrt{\kappa}. 
\end{equation*}
The degree of freedom \(\hat{\delta}\)  of \(Z\) is given by,
 \begin{equation*}
 \hat{\delta}=\frac{4\hat{a}}{\hat{c}^{2}}=\frac{4a}{c^{2}},
 \end{equation*}
 which is the same as the original degree of freedom parameter \(\delta\) for the CIR process \(X\).
 That means for the CIR processes \(X\) and \(Z\), when considering the simulation based on the parameter sets
 $(\delta,b,c)$ and $\bigl(\hat{\delta},\hat{b},\hat{c}\bigr)$,
 the only difference is in the diffusion parameters \(c\) and \(\hat{c}\). Since \(\kappa=\hat{c}^{2}/c^{2}\), we have,
\begin{equation*}
\frac{Z_{t}-X_{t}}{\hat{c}-c}=\frac{\kappa X_{t}-X_{t}}{\hat{c}-c} =\frac{X_{t}(\sqrt{\kappa}+1)}{c}.
\end{equation*}
For \(\hat{c}\rightarrow c\),  we have \(\kappa=\frac{\hat{c}^{2}}{c^{2}}\rightarrow 1\), and thus the derivative of the CIR process with respect to parameter \(c\) is \(2X_{t}/c\).

\section{Applications}\label{sec application}
We apply our direct inversion method to the pricing of a standard European put option,
an Asian option, and a basket option to show the efficiency of our method for path-independent, path-dependent and multi-asset options.
Then we calibrate the European put option and the Asian option. We consider the model under the risk-neutral measure, and assume that the discount factor is equal to \(1\).
  
 \subsection{Path-independent options: a put option on an interest rate}
 As an example of a path-independent option, we price a put option on an interest rate that is modelled as a CIR process.
 First, we use our direct inversion method to simulate the squared Bessel process, which then generates the CIR process.
 The payoff of the put option with one unit notional amount is,
\begin{equation*}
\bigl(K-X_{T}\bigr)^+,
\end{equation*}
where \(K\) is the strike rate, and \(X_{T}\) is the interest rate at the expiration time \(T\).  
The price  \(C_{0}\) of the put option can be calculated as follows:
\begin{equation*}
C_{0}=\mathrm{E}\bigl[(K-X_{T})^+\bigr]=\int_{0}^{K} (K-x) f(x)\,\mathrm{d}x,
\end{equation*}
where \(f\) is the probability density function. The CIR process \(X\) is given as in Equation (\ref{eq:CIR}) 
and its density function can be expressed through the density function of the non-central chi-square distribution. 
Therefore, the exact price of the put option can be calculated. Let \(\hat{X}_{T}\) denote the approximated interest rate at the expiration time \(T\). Then 
\(\hat{C}_{0}\) can be written as the expectation, 
\[
\hat{C}_{0}=\mathrm{E}\bigl[(K-\hat{X}_{T})^+\bigr].
\]
The relative error of the approximation is defined as follows:
\[
e\coloneqq\frac{|C_{0}-\hat{C}_{0}|}{C_{0}}.
\]
We investigate the efficiency of our direct inversion method by considering the relative error, the CPU times, and the length of the timestep in the path simulation. 
As a comparison, we choose the full truncation Euler scheme by Lord, Koekkoek and Dijk \cite{lord2010comparison},
and the quadratic-exponential scheme with martingale-correction by Anderson \cite{andersen2008simple}.
The full truncation Euler method is developed from the Euler method to avoid negative values when generating samples of the CIR process.
The quadratic-exponential scheme with martingale-correction (QE-M) is an efficient and robust method to generate samples
of the CIR process for small values for the degrees of freedom (see details in Anderson \cite{andersen2008simple}), consistent with our case. 
 
Let $Y$ be 
the squared Bessel process as shown in Equation (\ref{squared Bessel SDE}) with values \(\delta=0.18\), and \(Y_{0}=0.09\).
Then for the CIR process $X$ given by Equation (\ref{eq:CIR}),
we set \(a=0.045\), \(b=-0.5\), \(c=1\), so that \(\delta=4a/c^{2}=0.18\).  Recall $W$ is the Wiener process defined by 
\(W_{t}=2\int_{0}^{t}(\sqrt{\mathrm{e}^{bs}}/c)\, \mathrm{d}B_{C(s)}\), where \(C(s)=\left(c^{2}/4b\right)\cdot \left(1-\mathrm{e}^{-bs}\right)\).
The expiry time of the option is \(T=10\), the initial value of the asset is \(X_{0}=0.09\), and the strike price is \(K=0.09\).

We use an equidistant partition of the time interval $[0, T]$
into $N$ timesteps of length $h:=T/N$.  
We show the relationship between the length $h$ of the timesteps and the relative errors for the full truncation Euler method,
the QE-M method, and direct inversion method, respectively, in Figure \ref{figlogerrorx}. The relative errors decrease with decreasing $h$ for the full truncation Euler method and QE-M method. 
For the direct inversion method, there is no need to discretize, and only one timestep is needed to simulate one asset path if the option is path-independent. Hence
for our direct inversion method, the relative errors are kept
within the order of \(10^{-3}\) 
for the different lengths of timesteps, showcasing the efficiency
of this method for path-independent options.

In Table \ref{talbex}, we show the CPU times and the maximum length of the timesteps for each of the three methods in order 
to achieve  at least order 
\(10^{-3}\) accuracy in pricing the put option.
As shown in Table \ref{talbex}, to achieve this accuracy, the direct inversion method is the most efficient in terms of CPU times.
 An important reason for this is the large number of timesteps used in the simulation of the CIR process  in the full truncation Euler approximation to achieve the required accuracy. 
 The length of timestep is \(1/10\) for the full truncation Euler method, and \(1/4\) for the QE-M scheme.
 However, as noted above our direct inversion method does not require  discretization and can be performed in a single timestep over the whole time interval $[0, T]$.

\begin{figure}[H] 
\centering 
\includegraphics[width=0.65\textwidth]{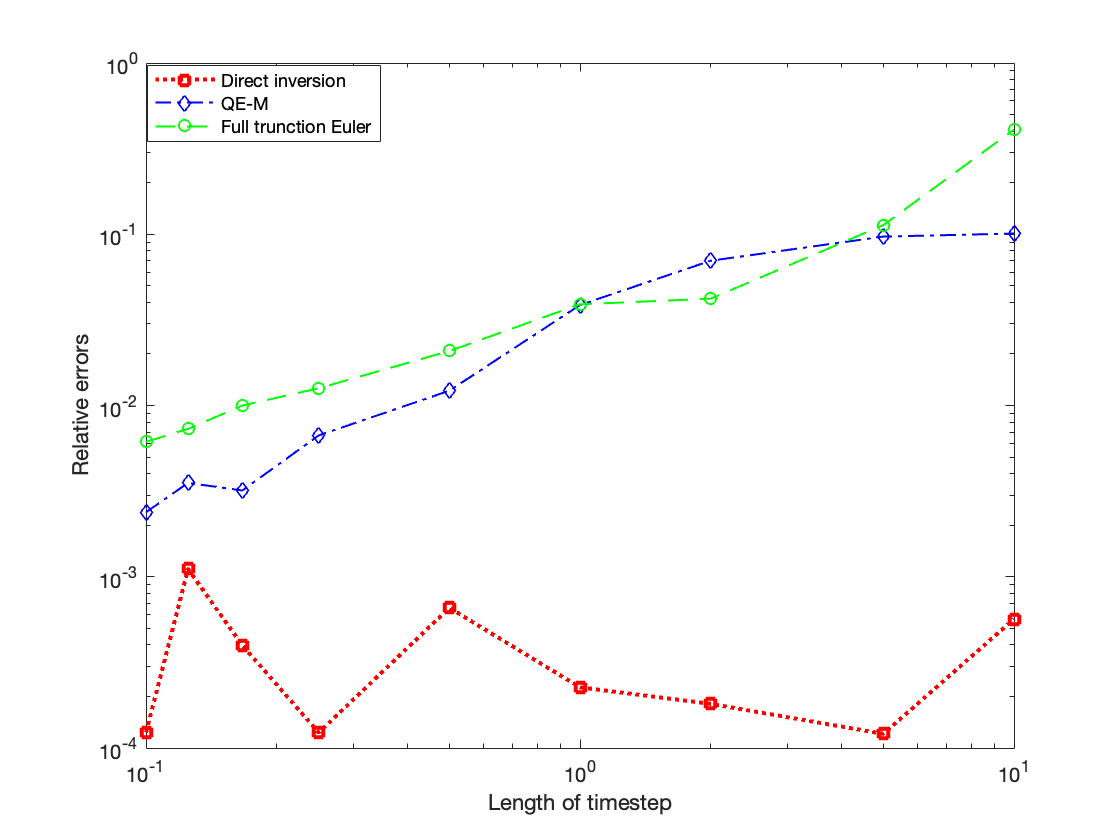} 
\caption{We show the relationship between the relative errors and the length of the timesteps in each simulation path 
for the full truncation Euler method, the QE-M method, and the direct inversion method. The number of simulation paths is \(10^{6}\). }
\label{figlogerrorx} 
\end{figure}
 
\begin{table}[H]
\caption{We show the maximum length of timesteps in each simulation path and the CPU times (in seconds) to price the put option with approximation accuracy 
\(O(10^{-3})\) for the full truncation Euler method,  the QE-M method, and the direct inversion method. We simulate \(10^{6}\) paths in each case. 
}
\label{talbex} 
\centering
\begin{tabular}{|l|l|l|l|}
\hline
               & Full Truncation Euler              & QE-M                          & Direct-Inversion              \\ \hline
Relative Error & \(6.14\times 10^{-3}\) & \(3.85\times 10^{-3}\) & \(3.12\times 10^{-3}\) \\ \hline
Length of timestep 
\(h\)& \(1/10\)                         & \(1/4\)                          & \(10\)                             \\ \hline
CPU time       & \(28.38\)                         & \(6.14\)                          & \(0.42\)                          \\ \hline
\end{tabular}
\end{table}

\begin{Remark}
  As shown in Figure \ref{figlogerrorx},  when we simulate with \(10^{6}\) paths,
  the relative errors of the estimated option price by our direct inversion method is of order \(10^{-3}\), which is essentially caused by the Monte Carlo error.
  The data for the direct inversion method in Table \ref{talbex} and Figure \ref{figlogerrorx} is generated by the Chebyshev expansion with accuracy \(O(10^{-4})\). 
  For the two-dimensional Chebyshev coefficients to achieve the \(O(10^{-4})\) accuracy,
  we have the truncation orders \(M=1\) and \(N=4\) in the first region, \(M=2\) and \(N=5\) in the middle region, and \(M=2\) and \(N=7\) in the tail region.
  Comparing with Table \ref{coefficient}, the number of Chebyshev coefficients required to achieve \(O(10^{-4})\) accuracy is much smaller than that needed to achieve \(O(10^{-8})\) accuracy.
  In Figure \ref{chebyshev}, we show the relative errors and the CPU times of the direct inversion method using
  the two-dimensional Chebyshev expansion with accuracy \(O(10^{-4})\) and with accuracy \(O(10^{-8})\).
  The first three data points  (from left to right) are generated with \(10^{4}\), \(10^{5}\),
  and \(10^{6}\) number of paths, which lead to Monte Carlo errors of order \(10^{-2}, 10^{-2.5}\), and \(10^{-3}\).
  In these cases, the Chebyshev approximations with accuracy \(O(10^{-4})\) and with accuracy \(O(10^{-8})\) have comparable 
  relative errors, whereas the Chebyshev approximation with accuracy of order
  $10^{-8}$ requires more CPU time. Since in these cases the Monte Carlo errors dominates the relative errors, the 
  higher-accuracy Chebyshev expansion will not have a significant effect on the relative errors.
  Therefore, in these situations we can reduce the accuracy of the Chebyshev expansion to the same order as the Monte Carlo errors resulting in 
  less CPU time while retaining the approximation accuracy.
  The last two data points (from left to right) are generated with \(10^{7}\) and \(10^{8}\) number of paths, which lead to Monte Carlo errors of order \(10^{-3.5}\) and \(10^{-4}\), respectively. 
  For the Chebyshev approximation with accuracy \(O(10^{-4})\), the relative errors remain at level \(O(10^{-3})\) 
  while for the Chebyshev approximation with error \(O(10^{-8})\), the relative errors decrease further to about \(O(10^{-4})\). 
  Hence to achieve an efficient approximation, the accuracy of the Chebyshev approximation should be chosen to match the relative Monte Carlo error.
\end{Remark}

\begin{figure}[H] 
\centering 
\includegraphics[width=0.65\textwidth]{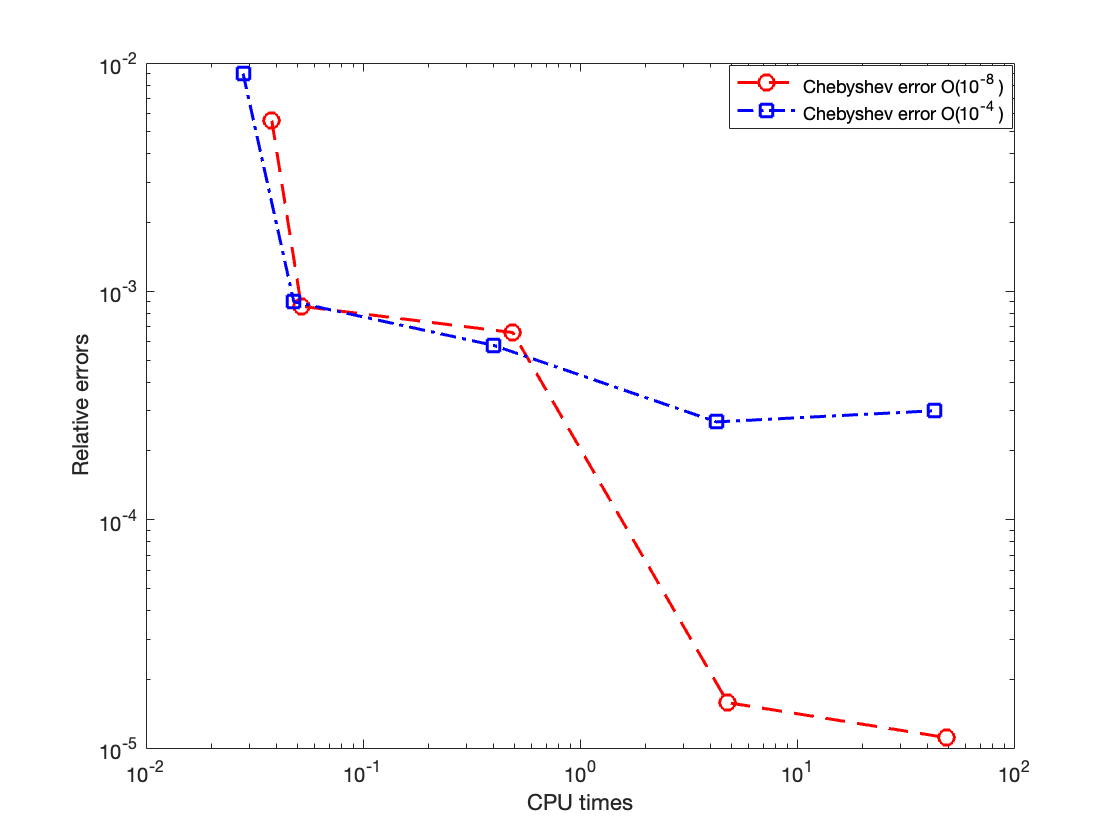} 
\caption{ We compare the relative errors and CPU times
 (in seconds) of the direct inversion method using the two-dimensional Chebyshev expansion with accuracy \(O(10^{-4})\) and \(O(10^{-8})\).
  The number of paths in each simulation is \(\left[10^{4}, 10^{5}, 10^{6}, 10^{7}, 10^{8}\right]\) (from left to right)
  which lead to Monte Carlo errors of order \(\left[10^{-2}, 10^{-2.5}, 10^{-3}, 10^{-3.5}, 10^{-4}\right]\).  }
\label{chebyshev} 
\end{figure}

\subsection{Path-dependent options: an Asian option}
We apply our direct inversion method to the pricing of a path-dependent option. Here we consider an Asian option with a payoff function as follows: 
\begin{equation}\label{eq:Asianoption}
\biggl(K-\frac{1}{M}\sum_{m=1}^{M}X_{t_{m}}\biggr)^+,
\end{equation}
where \(K\) is the strike price, and \(X_{t_{m}}\) is the asset price at the monitor times \(0=t_{0},t_{1},\cdots,t_{M}=T\), and \(t_{m}=mT/M\).
As above, to simulate the CIR process \(X\) we use our direct inversion method to simulate the squared Bessel process first,
and then use the simulated squared Bessel process to generate the CIR process as described in equation
\eqref{eq:CIR}.

We apply the full truncation Euler method, the QE-M method, and the direct inversion method to the pricing of the Asian option with yearly fixings and with quarterly fixings.
The parameter values are same as for the put option in Section 4.1.
For both cases of yearly fixings and quarterly fixings, 
the direct inversion method achieves high acuracy with a standard deviation of order $10^{-5}$.
The CPU time required compares favourably with the CPU times of the full truncation Euler method and of the QE-M method for the same level 
of accuracy (shown in the last line for each case in Table \ref{tbl:Asian}).
Similar to the case of the put option considered above,
the reason for this is that the direct inversion method does not require any discretization other than at the monitoring times,
while the full truncation Euler method and the QE-M method require a finer sampling rate.
For the direct inversion method the monitoring frequency in a path-dependent option determines the times at which the path must be simulated. Thus 
the required monitoring frequency 
may have an effect on the efficiency of the direct inversion method.

\begin{table}[H] 
  \caption{We present the estimation results of the full truncation Euler method,
    the QE-M method and our direct inversion method using \(10^{6}\) paths for an Asian option with yearly fixings and quarterly fixings.
    In each case the triple shown is the estimated price, the sample standard deviation (inflated by \(10^{3}\)) and CPU time, in seconds.
    We test the full truncation Euler methods and the QE-M method with different lengths of timestep \(h\).
    However, for the direct inversion method, the length of the timestep does not influence the approximation accuracy.  
}

\centering
\label{tab:my-table}
\begin{tabular}{|l|l|l|l|}
\hline
\multicolumn{1}{|l}{Stepsize} &\multicolumn{1}{l}{Full Truncation Euler}&\multicolumn{1}{l}{QE-M} & \multicolumn{1}{l|}{Direct inversion} \\ \hline
\multicolumn{4}{|l|}{M=10}                                                          \\ \hline
\multicolumn{1}{|l}{h=1}       &\multicolumn{1}{l}{\([0.0366, 0.0473, 6.98]\)} & \multicolumn{1}{l}{\([0.0469, 0.0358, 1.12]\)}&\([0.0464, 0.0341, 4.49]\) \\ 

\multicolumn{1}{|l}{h=1/2}& \multicolumn{1}{l}{\([0.0416, 0.0349, 13.29]\) } & \multicolumn{1}{l}{\([0.0463, 0.0347, 2.58]\) }  & \\
\multicolumn{1}{|l}{h=1/4}   &\multicolumn{1}{l}{\([0.0445, 0.0343, 26.12]\)}&\multicolumn{1}{l}{\([0.0464, 0.0341, 6.88]\)} &                                      \\ \hline
\multicolumn{4}{|l|}{M=40}                                                                \\ \hline
\multicolumn{1}{|l}{h=1/4}  & \multicolumn{1}{l}{\([0.0416, 0.0331, 25.89]\)}     &\multicolumn{1}{l}{\([0.0446, 0.0328, 6.28]\)}     &\([0.0444, 0.0323, 17.89]\)                               \\
\multicolumn{1}{|l}{h=1/8}  & \multicolumn{1}{l}{\([0.0432, 0.0326, 51.22]\)}     &\multicolumn{1}{l}{\([0.0445, 0.0325, 15.89]\)}     &                                  \\
\multicolumn{1}{|l}{h=1/16} & \multicolumn{1}{l}{\([0.0446, 0.0325, 104.62]\)}     &\multicolumn{1}{l}{\([0.0444, 0.0323, 36.28]\)} &                             \\ \hline
\end{tabular}%
\label{tbl:Asian}
\end{table}

\subsection{Basket option}
We consider a basket option with payoff function, 
\begin{equation*}
\biggl(K-\sum_{i=1}^{d}w_{i}X^{(i)}\biggr)^+,
\end{equation*}
where \(d\) is the number of assets underlying the basket option , \(w_{i}\) is the weight of each asset, and the underlying asset \(X^{(i)}\) is given as a CIR process,
\[
\mathrm{d}X_{t}^{(i)}=\bigl(a_{i}+b_{i}X_{t}^{(i)}\bigr)\mathrm{d}t+c_{i}\sqrt{X_{t}^{(i)}}\mathrm{d}W_{t}^{(i)},
\]
with constants \(a_{i}, b_{i}, c_{i}\) for \(i=1,\cdots, d\).
We consider both identical and independent scenarios of Brownian motion \(W_{t}^{(i)}\), which will be discussed below.

Recall  $X$ can be expressed as a scaled and time-changed Bessel process \(X_{T}^{(i)}=e^{b_{i}T}Y^{(i)}_{C_{i}(T)}\), 
where $Y^{(i)}$ is  a squared Bessel process, and where the function \(C_{i}(\cdot)\) is given in equation (\ref{ct}).
We have, 
\(Y^{(i)}_{C_{i}(T)}=\left(C_{i}(T)-C_{i}(0)\right)\cdot \chi_{\delta_{i}}^{2}(\lambda_{i})\) with \(\delta_{i}=4a_{i}/c_{i}^{2}\) and \(\lambda_{i}=Y^{(i)}_{C_{i}(0)}/\left(C_{i}(T)-C_{i}(0)\right)\). 
We use our direct inversion method to price the basket option for different parameter values  \(a_{i}\) \(b_{i}\), and \(c_{i}\).
The values for \(a_{i}\), \(b_{i}\), and \(c_{i}\) are shown in  Table \ref{basketoption}. 
  We  consider different values of \(\delta_{i}\), \(b_{i}\), and \(c_{i}\). Since \(a_{i}=\delta_{i}\cdot c_{i}^{2}/4\), here the parameter \(a_{i}\) is treated as a fixed parameter that
depends on the values of \(\delta_{i}\) and \(c_{i}\). 
We set \(T=10\), \(K=0.09\), \(d=5\), \(X_{0}^{(i)}=0.09\), \(Y_{C(0)}^{(i)}=0.09\), and \(w_{i}=1/5\) for \(i=1,\cdots, 5\).

To simulate the CIR process,
 we need to generate the non-central chi-square samples \(\chi_{\delta_{i}}^{2}(\lambda_{i})\). 
Recall we have \(\chi_{\delta_{i}}^{2}(\lambda_{i})=\chi^{2}_{\delta_{i}}+\chi^{2}_{2N_{i}}\) where \(N_{i}\) is a Poisson random variable with mean \(\lambda_{i}/2\). 
The first central chi-square sample \(\chi^{2}_{\delta_{i}}\) is generated by our direct inversion method using common random numbers for all assets in the basket option,
i.e., we generate \(\chi^{2}_{\delta_{1}},\cdots, \chi^{2}_{\delta_{5}}\) from the same Uniform\([0,1]\) random variables.
The second central-chi square sample \(\chi^{2}_{2N_{i}}\) is generated with the method introduced in Malham and Wiese \cite{r2} and described in Section \ref{ssec:noncentral}. 
For assets $X^{(i)}$ with parameter combinations 
$(a_i, c_i)$ that result in the same degrees of freedom $\delta= \delta_i=
4a_i/c_i^2$,  
we only need to generate one central chi-square variable
\(\chi_{\delta}^{2}\) that can be used for all of these assets.
 For example, 
 in cases \(1\) and \(2\), we have a fixed value of the degrees of freedom \(\delta_{i}\equiv0.18\),
 which means we generate one central chi-square sample \(\chi^{2}_{\delta_{i}}\) only and use it in the simulation of all five assets in the basket option.
 In cases \(3\) and \(4\), the value of the degrees of freedom \(\delta_{i}\) for each asset is different. Hence, we generate new central chi-square samples 
\(\chi^{2}_{\delta_{i}}\) for each asset, resulting in longer CPU times as shown in Table \ref{basketoption}.
 The CPU times in cases 1 and 2 are nearly the same, whereas the value of parameter \(b_{i}\) varies. 
 This indicates our method is robust and not affected by the choice of value for the parameter \(b_{i}\).
 By comparing cases 3 and 4, we also see that the change in value of parameter \(c_{i}\) does not affect the efficiency of our direct inversion method. 

 For comparison with our direct inversion method we choose the most favourable case
 for the full truncation Euler method of using the same driving Brownian motion for all five assets underlying the basket option.
 We choose the discretization timestep \(h=1/10\) for the full truncation Euler method to simulate the assets \(X_{t_{1}}^{(i)},\cdots, X_{T_{N}}^{(i)} \) with \(N=T/h\)
 to achieve an approximation accuracy comparable with our direct inversion method. The prices of the basket options
 are slightly different when using the full truncation Euler method compared with the direct inversion method.
 This is because the same seeds are used to generate the random variables driving the asset price processes in the full truncation Euler method.
 For the direct inversion method, only the first central chi-square sample \(\chi^{2}_{\delta_{i}}\) are generated by the same seeds;
 the second central chi-square sample \(\chi^{2}_{2N_{i}}\) cannot be generated with the same seeds because different values of parameters will lead to different 
means of the Poisson variables \(N_{i}\). 
The CPU times in the five cases are nearly identical for the full truncation Euler method. Because the asset processes have different parameter values,
every asset path needs to be generated by the full truncation Euler method (although with the same seed).  
 In all five cases, our direct inversion method compares favourably in the CPU times, especially for case 1 and 2 with fixed degrees of freedom \(\delta_{i}\).
 
\begin{table}[H]
  \caption{We present the estimation results of the full truncation Euler method and the direct inversion method using \(10^{6}\) paths for a basket option with five assets.
    In each case the triple shown is the estimated price, the sample standard deviation (inflated by \(10^{3}\)) and CPU times (in seconds).  
}

\centering
\label{basketoption}
\begin{tabular}{l|l}
\hline\hline
                             \multicolumn{2}{l}{  Case 1    }                                                                    \\ \hline
Degree of freedom \( \delta\)  &   \( \delta_{i}=0.18 \),    for \(i=1,\cdots, 5\) \\ \hline
Free parameters            & \begin{tabular}[c]{@{}l@{}}\( b_{i}=-0.5 \),    for \(i=1,\cdots, 5\)\\\([c_{1}, c_{2}, c_{3}, c_{4},c_{5}] =[0.8,0.9,1,1.1,1.2]\)\end{tabular}        \\ \hline
Fixed parameter & \([a_{1},a_{2},a_{3},a_{4},a_{5}]=[0.0288,   0.0365,  0.0450,  0.0545,   0.0648] \)\\\hline

Full truncation Euler method  &  \([0.0666, 0.0350, 10.46]\)                                                                                 \\ \hline
Direct inversion method       &\([0.0690,0.0344,1.81] \) \\ \hline\hline
     
\multicolumn{2}{l}{  Case 2    }                                                                    \\ \hline
Degree of freedom \( \delta\)  &   \( \delta_{i}=0.18 \),    for \(i=1,\cdots, 5\) \\ \hline
Free parameters            & \begin{tabular}[c]{@{}l@{}} \([b_{1},b_{2},b_{3},b_{4},b_{5}]=[-0.4,-0.45,-0.5,-0.55,-0.6] \)\\ \([c_{1}, c_{2}, c_{3}, c_{4},c_{5}] =[0.8,0.9,1,1.1,1.2]\) \end{tabular}        \\ \hline
Fixed parameter & \([a_{1},a_{2},a_{3},a_{4},a_{5}]=[0.0288,   0.0365,  0.0450,  0.0545,   0.0648] \) \\\hline
Full truncation Euler method  &  \([0.0666, 0.0349, 10.39]\)                                                                                 \\ \hline
Direct inversion method       &\([0.0692,0.0344,1.90] \) \\ \hline\hline

\multicolumn{2}{l}{ Case 3        }                                                                \\ \hline
Degree of freedom \( \delta\)  &   \( [\delta_{1},\delta_{2},\delta_{3},\delta_{4},\delta_{5}] =[0.1152,   0.1458,  0.1800,  0.2178,   0.2592 ]
\)    \\ \hline

Free parameters & \( b_{i}=-0.5 \), \(c_{i}=1\),    for \(i=1,\cdots, 5\)\\\hline
Fixed parameter           &  \([a_{1},a_{2},a_{3},a_{4},a_{5}]=[0.0288,   0.0365,  0.0450,  0.0545,   0.0648]\)        \\ \hline

Full truncation Euler method  &  \([0.0642, 0.0354, 10.65]\)                                                                                 \\ \hline
Direct inversion method       &\([0.0676,0.0346, 2.59] \) \\ \hline\hline
\multicolumn{2}{l}{Case 4  }  
\\ \hline
Degree of freedom \( \delta\)  &   \( [\delta_{1},\delta_{2},\delta_{3},\delta_{4},\delta_{5}] =[ 0.2250,   0.2000, 0.1800,  0.1636,   0.1500]\)    \\ \hline

Free parameters            & \begin{tabular}[c]{@{}l@{}} \( b_{i}=-0.5 \),    for \(i=1,\cdots, 5\)\\ \([c_{1},c_{2},c_{3},c_{4},c_{5}]=[0.8,0.9,1,1.1,1.2] \)  \end{tabular}      \\ \hline
Fixed parameter & \(a_{i}=0.045\),   for \(i=1,\cdots, 5\)\\\hline

Full truncation Euler method  &  \([0.0649, 0.0350, 10.76]\)                                                                                 \\ \hline
Direct inversion method       &\([0.0679,0.0344, 2.69] \) \\ \hline
\end{tabular}
\end{table}
 
\subsection{Calibration of the CIR Model}\label{chap6}

In this section we apply the direct inversion method to achieve a robust calibration of the CIR model.
The advantage of the direct inversion method is not only in the efficiency in asset path simulation but also in the flexibility with respect to the parameter values.
Specifically, the direct inversion method is based on a two-dimensional Chebyshev expansion, where  one of the variables  is the degree of freedom parameter of the squared Bessel process.
This means only one series of coefficients is required when simulating paths of the asset price process with different values of degree of freedom parameter.
This flexibility in parameter values is essential in model calibration.
The calibration is achieved by solving a nonlinear least-squares problem using the Levenberg-Marquardt method.
The analytical gradients that we use for the Levenberg-Marquardt method are deduced from the expressions of the sensitivities
of the CIR model with respect to its model parameters given in Section \ref{section4.2}. In this application section,
we calibrate synthetic data generated from the CIR model. Hence, in other words, the calibration is a parameter-recovery test.

\paragraph{Calibration problem and Levenberg-Marquardt method.}
The calibration of stochastic models to observed market data is a fundamental problem in quantitative finance.
The objective of the model calibration is to determine the values of the model parameters to match the market prices.
Here we consider the calibration of option prices, and assume that the underlying asset price process is a CIR process.
The calibration problem can be formulated as an optimization problem, in which the errors between model prices and market prices are minimized.
The nonlinear least-squares method is a common approach to solving this optimization problem. 
We calibrate the model to European put option prices and Asian option prices by finding the optimal values of the parameters using the Levenberg-Marquardt method \cite{more2006levenberg}.

Let \(P^{*}(K_{i},T_{i})\) denote the market price of an option with strike price \(K_{i}\) and maturity time \(T_{i}\), and \(P(\theta; K_{i},T_{i})\)
be the option price calculated by the CIR model satisfying Equation (\ref{eq:CIR})  with the parameter vector \(\theta\coloneqq(a,b,c)^{T}\).
To calibrate the CIR model, we minimize the difference between the option price calculated by the CIR model and the observed data in the real market.
The difference between the model data and the real market data is measured by the residual.
Here we set the residual vector for \(n\) options as follows:
\begin{equation*}
    r(\theta)\coloneqq\left(r_{1}(\theta),r_{2}(\theta),\cdots,r_{n}(\theta)\right)^{T},
\end{equation*}
where, 
\begin{equation*}
    r_{i}(\theta)=P(\theta;K_{i},T_{i})-P^{*}(K_{i},T_{i})
\end{equation*}
with \(i=1,2,\cdots,n\).
We define the function $f$ of the residuals as follows:
\begin{equation*}
    f(\theta)\coloneqq\tfrac12r^{T}(\theta)r(\theta),
\end{equation*}
Calibration of the CIR model can be treated as a nonlinear least-square problem; indeed  
the $\theta$-value minimizing $f(\theta)$ can be found via standard nonlinear optimization methods. 

A widely used method to solve nonlinear least-squares problem is the Levenberg-Marquardt method
(see Levenberg \cite{levenberg1944method}, Marquardt \cite{marquardt1963algorithm} and see details in More \cite{more2006levenberg}).  
This method can be viewed as a combination of two typical minimization techniques, the Gauss-Newton method and the gradient descent method, respectively.
For the search step of the Levenberg-Marquardt method, only the first derivatives of the target parameters are required.
We have derived the first derivatives of the CIR process with respect to its three parameters in Section \ref{section4.2}.
These results can be used to obtain analytical gradients for the Levenberg-Marquardt method. Analytical gradients are more accurate and efficient than numerical gradients.
The Levenberg-Marquardt method to find the optimal values of a parameter $\theta$ is as follows:
$\theta_{k+1}=\theta_{k}+\Delta \theta_{k}$,
where \(k=0,1,2,\ldots\).
The search step is defined as
  $ \Delta\theta=(JJ^{T}+\mu I)^{-1}\nabla f $,
where \(I\) is the identity matrix, and \(\mu\) is a damping factor.
Here \(J=\nabla r^{T}(\theta)\) is the Jacobian matrix of the residual vector \(r(\theta)\).
For an option with \(i=1,2,\ldots, n\), and the specific parameter for \(j=1,2,3\), we have  the Jacobian matrix,
 \begin{equation*}
    J_{ji}=\partial_{\theta_{j}} r_{i}(\theta)=\partial_{\theta_{j}} P(\theta; K_{i},T_{i}). 
\end{equation*}
Then the gradient of the residual function is $ \nabla f=Jr(\theta)$.

In the following sections, we perform the calibration of the CIR model to European put options and Asian options using the Levenberg-Marquardt method.

\paragraph{Option calibration}
We list the parameter values in Table \ref{parameter}. Here we have the parameter set \(\theta^{*}=(\delta^{*},b^{*},c^{*})\),
which are real parameter values for the data used in the calibration.
The value of parameter \(a^{*}\) depends on parameter \(\delta^{*}\) based on the relationship between the CIR process and the squared Bessel process.
For the combination of different expiration time \(T_{i}\) and strike price \(K_{i}\), we generate synthetic data to represent the real market prices in our numerical tests.
Other sets of model parameters have also been tested, and all of the tests have similar accuracy as the following examples.
\begin{table}[H]
\caption{Parameters for the European put option and Asian option. 
}
\label{parameter} 
\centering
\begin{tabular}{|cr|cr|}
\hline 
\multicolumn{2}{|l|}{Model parameters }            & \multicolumn{2}{|l|}{Option specifications }                      \\ \hline
\(\delta^{*} \)& \(0.18\) & \(X_{0}\) &\(0.09\)  \\ 
\(a^{*}\) & \(0.045\)                         &\(K_{i}\)& \(\quad\quad[0.065,0.085,0.090,0.1,0.11,0.12]\)                              \\ 
\(b^{*}\)     & \(-0.50\)                         &\(T_{i}\)&\([6.00,7.00,8.00,9.00,10.00]\)
\\ 
\(c^{*}\)&  \(1.00\)  &   & \\ \hline
\end{tabular}
\end{table}

First, we consider the calibration to the European put option.
To apply the Levenberg-Marquardt method, the gradient for each parameter is required. For the European put option, the gradient used in the calibration is as follows: 
\begin{equation*}
  \partial_{\theta_{j}} P(\theta; K_{i},T_{i})
  =\partial_{\theta_{j}}  E\Bigl[\bigl(K_{i}-X_{T_i}\bigr)^{+}\Bigr]
  =E\Bigl[-\partial_{\theta_{j}}  X_{T_i}\cdot 1_{\{X_{T_i}\leq K_{i}\}}\bigl(X_{T_i}\bigr)\Bigr],  \label{2.8}
\end{equation*}
where the quantities $\partial_{\theta_{j}}  X_{T_i}$ are computed as outlined in sections \ref{sec4.11} and \ref{section4.2}.
Hence the analytical gradient of the option price can be approximated
by simulating both the asset price $X_{T_i}$ and the quantities $\partial_{\theta_{j}} X_{T_i}$.

The summary of the implementation of calibration for the put option is in Table \ref{put_cal}, in which we show the average results of 100 tests.
The initial guess \(\theta_{0}=\left(\delta_{0},b_{0},c_{0}\right)\) is randomly selected in the \(\pm 20\% \) range of the real parameter set \(\theta^{*}=\left(\delta^{*},b^{*},c^{*}\right)\).
The parameter set \(\theta^{+}=(\delta^{+}, b^{+}, c^{+})\) is the optimal parameter values deduced from the calibration,
and \(r_{+}\) is the residual of the optimal parameters, and \(\Delta\theta^{+}\) is the values of final search step.
The deviations of the optimal values and the real model values for the four parameters are within \(O(10^{-2})\),
which are in reasonable ranges compared to the calibration results in other literature
(for example, Liu, Borovykh, Grzelak, and Oosterlee \cite{liu2019neural}).
After about \(8\) iterations, the residuals of the model price and the real option price are within \(O(10^{-4})\).
As we use \(10^{7}\) paths in asset simulation, the calibration error is mainly caused by the Monte Carlo error.

\begin{table}[H]
\caption{We give the calibration results of the European put option. These are average results of 100 test cases. }
\label{put_cal} 
\centering
\begin{tabular}{|ll|ll|}
\hline \hline
\multicolumn{2}{|l}{Absolute deviation from \(\theta^{*}\) }            & \multicolumn{2}{|l|}{Calibration summary    }                      \\ \hline
\(|\delta^{+}-\delta^{*}| \)& \(4.99\times 10^{-4}\) & \(\|r_{0}\|\) &\(6.88\times 10^{-2}\)  \\ \hline
\(|a^{+}-a^{*}|\) & \(1.69\times 10^{-3}\)                         &\(\|r_{+}\|\)& \(\leq 1\times 10^{-4}\)                              \\ \hline
\(|b^{+}-b^{*}|\)     & \(9.56\times 10^{-3}\)                         & \(\|\Delta \theta^{+}\|\)     &  \(8.32\times 10^{-3}\)                                        \\ \hline
\(|c^{+}-c^{*}|\)& \(1.92\times 10^{-2}\)   &  Iterations & \(8.2\) 
\\
\hline\hline
\end{tabular}
\end{table}

In Figure \ref{fig_put_cal}, we show the iteration information of the calibration for a representative example with initial guess \(\theta_{0}=(0.12,-0.80, 0.60)\).
The convergences of the deviations for parameters \(\delta\), \(a\), \(b\), and \(c\) are shown in the left figure.
The residual decrease from \(O(10^{-1})\) to  \(O(10^{-4})\) during the \(8\) iterations in the Levenberg-Marquardt method.

\begin{figure}[H] 
\centering 
\includegraphics[width=0.85\textwidth]{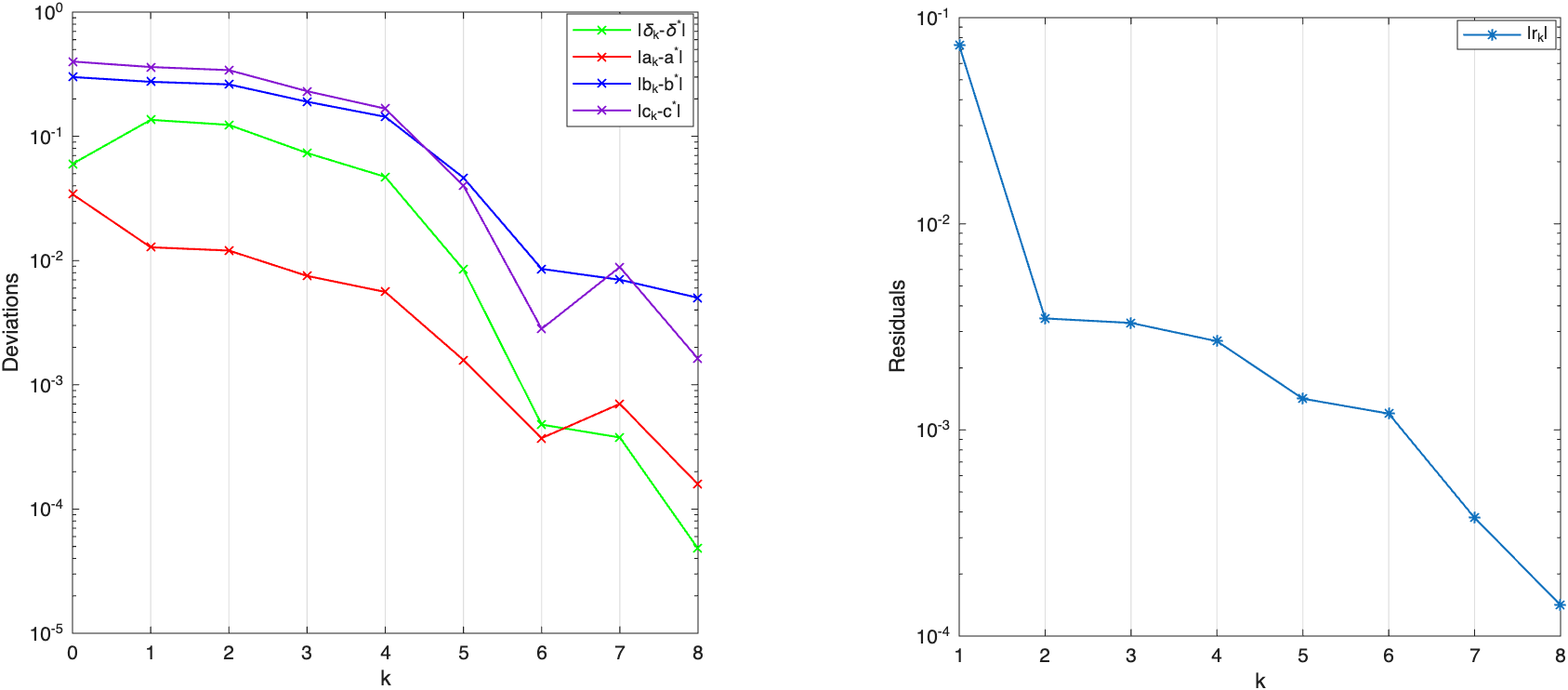} 
\caption{We give a representative example for the calibration of the European put option to show the convergence of the Levenberg-Marquardt method
  for the deviations of parameters \(\delta\), \(a\), \(b\), and \(c\), and the residuals in each iteration step \(k\).}
\label{fig_put_cal} 
\end{figure}

Next we consider the problem to calibrate the model to Asian option prices with payoff function given in Equation (\ref{eq:Asianoption}).
The derivative of the option price with respect to the parameters, which is the gradient used in the calibration, is as follows:
\begin{equation*}
  \partial_{\theta_{j}} P(\theta; K_{i},T_{i})
  =\partial_{\theta_{j}} E\biggl[\Bigl(K_{i}-\frac{1}{M_i}\sum_{m=1}^{M_i}X_{t_{m}}\Bigr)^{+}\biggr]
  =E\biggl[-\frac{1}{M_i}\sum_{m=1}^{M_i}\partial_{\theta_{j}} X_{t_{m}}\cdot 1_{\{\bar{X}_{T_i}\leq K_{i}\}}\bigl(\bar{X}_{T_i}\bigr)\biggr], 
\end{equation*}
where $\bar{X}_{T_i}:=\frac{1}{M_i}\sum_{m=1}^{M_i}X_{t_{m}}$.
Similarly to the European put option case,
the analytical gradients of the Levenberg-Marquardt method for the calibration of the Asian option are approximated
by simulating the asset path $X_{t_{m}}$ and the parameter sensitivities $\partial_{\theta_{j}} X_{t_{m}}$.

\begin{table}[H]
\caption{We show the calibration results of the Asian put option. These are average results of 100 test cases. 
}
\label{asian_cal} 
\centering
\begin{tabular}{|ll|ll|}
\hline \hline
\multicolumn{2}{|l}{Absolute deviation from \(\theta^{*}\) }            & \multicolumn{2}{|l|}{Calibration summary    }                      \\ \hline
\(|\delta^{+}-\delta^{*}| \)& \(9.99\times 10^{-4}\) & \(\|r_{0}\|\) &\(2.98\times 10^{-2}\)  \\ \hline
\(|a^{+}-a^{*}|\) & \(9.01\times 10^{-4}\)                         &\(\|r_{+}\|\)& \(3.99\times 10^{-4}\)                              \\ \hline
\(|b^{+}-b^{*}|\)     & \(1.39\times 10^{-2}\)                         & \(\|\Delta \theta^{+}\|\)     &     \(4.32\times 10^{-3}\)                                     \\ \hline
\(|c^{+}-c^{*}|\)&  \(1.23\times 10^{-2}\)  &  Iterations & \(14.4\)\\
\hline\hline
\end{tabular}
\end{table}

The calibration of the Asian option is tested \(100\) times with different initial guesses of the parameter set \(\theta_{0}\) selected in the \(\pm 20\%\) range of the parameter set \(\theta^{*}\).
We show the average results of \(100\) tests in Table \ref{asian_cal}. 
We achieve the same calibration accuracy for the Asian option and the European put option, which is shown in the same order of the deviations and final residuals.
However, the average calibration iterations for the Asian option are \(14.4\), while the European put option only requires \(8\) iteration steps.
In Figure \ref{fig_aisan_cal},  we show the iterations of the Asian option calibration with a representative example, in which the initial guess is \(\theta_{0}=(0.19,-0.80, 1.30)\).
Both the deviations and residuals convergence to reasonable orders in the iterations of the Levenberg-Marquardt method.

\begin{figure}[H] 
\centering 
\includegraphics[width=0.85\textwidth]{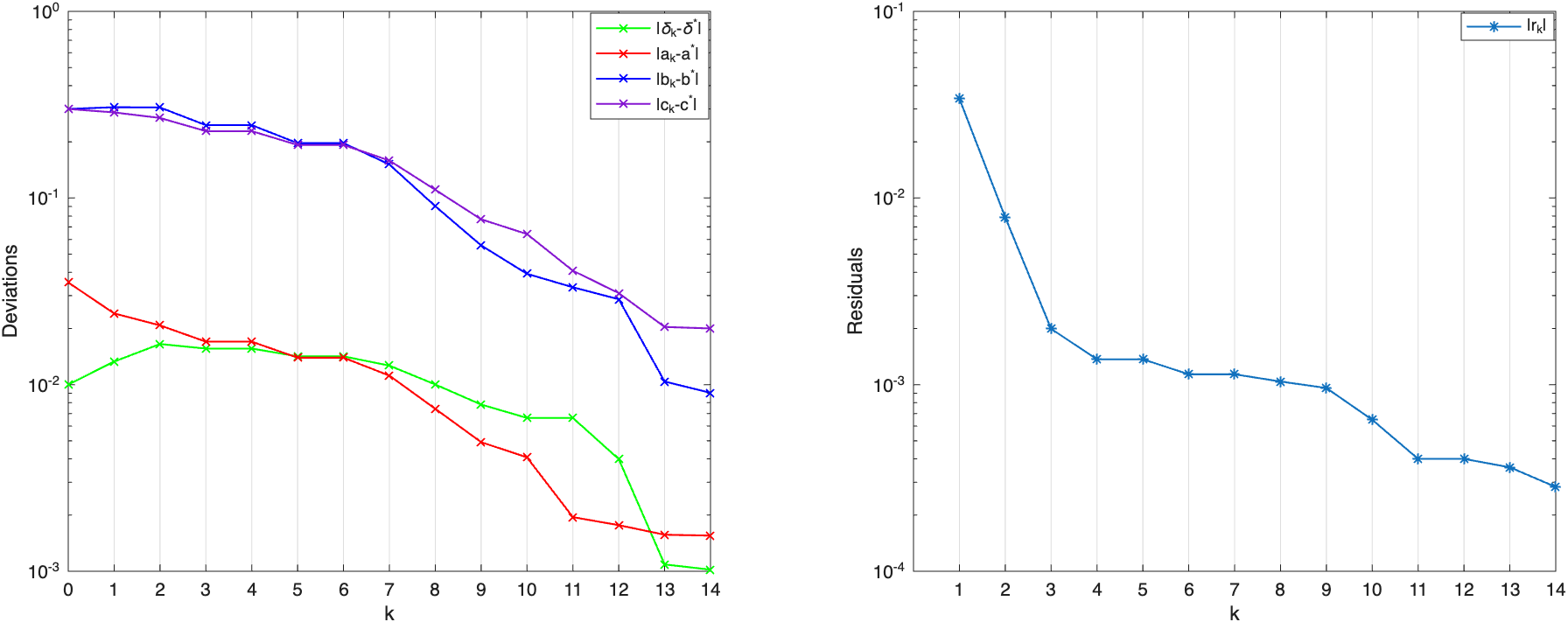} 
\caption{We give a representative example for the calibration of the Asian put option to show the convergence
  of the Levenberg-Marquardt method for the deviations of parameters \(\delta\), \(a\), \(b\), and \(c\), and the residuals in each iteration step \(k\).}
\label{fig_aisan_cal} 
\end{figure}

\section{Conclusion}
In this paper, we designed a new direct inversion method to generate non-central chi-square samples and hence to simulate the squared Bessel process.
The method is based on a two-dimensional Chebyshev approximation of the inverse distribution function, where the degrees of freedom of the chi-square sample are treated as one of the variables.
Hence we can use the same series of Chebyshev coefficients to generate chi-square samples for a range of values for the degrees of freedom, making our method particularly efficient in such a setting.
The direct inversion method is efficient and of high accuracy in generating chi-square samples for all degrees of freedom,
particularly for the case of small degrees of freedom. High accuracy here means using a Chebyshev expansion with accuracy of order \(10^{-8}\).
In principle, our method can be extended to a higher order, such as the 
order \(10^{-10}\) accuracy Chebyshev expansion in Malham and Wiese \cite{r2}.
As a direct inversion method, the accuracy of the method 
does not depend on the length of the timestep in path simulations of the squared Bessel process, an important factor for the efficiency of the method. 
We derive expressions for the derivatives of the mean-reverting square-root process with respect to the three model parameters.
These enable us to employ a gradient-based method, the Levenberg-Marquardt method,
to develop a flexible calibration algorithm for the mean-reverting square-root process based on the new direct inversion method.
This algorithm can be used for the calibration of (discretely observed) path-dependent options, such as Asian options.

 We list some directions for future research. Firstly, since we have 
treated the degrees of freedom $\delta$, or equivalently the dimension of the squared Bessel process, as a variable in our Chebyshev expansion,
we can analyse the impact the model parameter $\delta$ has on the model and on the method. 
Secondly, we can use the simulated squared Bessel process to simulate other stochastic models, such as the CEV model. 
We can extend our two-dimensional Chebyshev expansion to the simulation of the Heston model,
see also Malham, Shen, and Wiese \cite{malham2021series} who established direct inversion algorithms for this model.
Also see Malham and Wiese \cite{malham2014efficient}. 
Thirdly, Makarov and Glew \cite{makarov2010exact} develop an exact simulation method for squared Bessel and CIR processes
based on the representation of the transition probability as a randomized  Gamma distribution. 
Our direct inversion method provides an efficient and accurate method for generating chi-square samples and thus gamma samples,
particularly in the challenging small degrees of freedom scenario.
It would be interesting to investigate the combination of our direct inversion method with their exact simulation for computational efficiency and robustness.
Fourthly, we can consider extending our direct inversion method from the one-dimensional squared Bessel process to its matrix-valued extension, the Wishart process.

\paragraph*{Acknowledgment.}
    We would like to thank the referees for their constructive suggestions, which helped to improve the original manuscript.
\paragraph*{Disclosure of Interest.}
There are no interests to declare.
\paragraph*{Funding.}
No funding was received.

\clearpage

\appendix
\section{Appendix}

\begin{table}[H]
\caption{The Chebyshev coefficients for the first region with \(u\in[0,0.8]\) and \(\delta\in[0.1,0.2]\) to achieve \(O(10^{-8})\) approximation accuracy.
}

\centering
\label{coefficient_region1}
\begin{tabular}{|c| r r r |}
\hline
		 & \multicolumn{1}{c}{\(0\)}		   &  \multicolumn{1}{c}{\(	1	\)} & \multicolumn{1}{c|}{\(2\)}   \\ \hline
 \(	0	\)  &  \(	0.0015173224201204 	 \)   &  \(	0.0002528722215717 	\) & \(	-0.0000114009960766 			 \)  \\
\(	1	\)  &  \(	0.0028185551650815 	 \)   &  \(	0.0004053700919312 	\) & \(	-0.0000259448708964 			 \)  \\
\(	2	\)  &  \(	0.0022599358279589 	 \)   &  \(	0.0001678268383677 	\) & \(	-0.0000289961611465 			 \)  \\
\(	3	\)  &  \(	0.0015680342698270 	 \)   &  \(	-0.0000700741494637 	\) & \(	-0.0000215492684447 			 \)  \\
\(	4	\)  &  \(	0.0009459043966500 	 \)   &  \(	-0.0002038468732949 	\) & \(	-0.0000031908343145 			 \)  \\
\(	5	\)  &  \(	0.0005000148204097 	 \)   &  \(	-0.0002183501508953 	\) & \(	0.0000155516275132 			 \)  \\
\(	6	\)  &  \(	0.0002342786936571 	 \)   &  \(	-0.0001632918130127 	\) & \(	0.0000245736522666 			 \)  \\
\(	7	\)  &  \(	0.0000986660876581 	 \)   &  \(	-0.0000957800587457 	\) & \(	0.0000226176037315 			 \)  \\
\(	8	\)  &  \(	0.0000378469221872 	 \)   &  \(	-0.0000462238395762 	\) & \(	0.0000151912638058 			 \)  \\
\(	9	\)  &  \(	0.0000133301455983 	 \)   &  \(	-0.0000188923951304 	\) & \(	0.0000079916328208 			 \)  \\
\(	10	\)  &  \(	0.0000043196173257 	 \)   &  \(	-0.0000067083849095 	\) & \(	0.0000034157501336 			 \)  \\
\(	11	\)  &  \(	0.0000012906319703 	 \)   &  \(	-0.0000020933533018 	\) & \(	0.0000012177333308 			 \)  \\
\(	12	\)  &  \(	0.0000003609030015 	 \)   &  \(	-0.0000005962987155 	\) & \(	0.0000003726832304 			 \)  \\
\(	13	\)  &  \(	0.0000000998616155 	 \)   &  \(	-0.0000001639159899 	\) & \(	0.0000001012588862 			 \)  \\ \hline
& \multicolumn{1}{c}{\(3\)}		   &  \multicolumn{1}{c}{\(	4	\)} & \multicolumn{1}{c|}{\(5\)}   \\ \hline
\(	0	\)  &  \(	0.0000009598467591 	 \)   &  \(	-0.0000001024227037 	\) & \(	0.0000000122727754 			 \)  \\
\(	1	\)  &  \(	0.0000020896532161 	 \)   &  \(	-0.0000002151292692 	\) & \(	0.0000000253707092 			 \)  \\
\(	2	\)  &  \(	0.0000026001630603 	 \)   &  \(	-0.0000002531116714 	\) & \(	0.0000000283544586 			 \)  \\
\(	3	\)  &  \(	0.0000030696202825 	 \)   &  \(	-0.0000003242049038 	\) & \(	0.0000000349527349 			 \)  \\
\(	4	\)  &  \(	0.0000025977096176 	 \)   &  \(	-0.0000003843685990 	\) & \(	0.0000000452799900 			 \)  \\
\(	5	\)  &  \(	0.0000007618438592 	 \)   &  \(	-0.0000003286102337 	\) & \(	0.0000000524580734 			 \)  \\
\(	6	\)  &  \(	-0.0000015917811532 	 \)   &  \(	-0.0000000962454381 	\) & \(	0.0000000424043844 			 \)  \\
\(	7	\)  &  \(	-0.0000030850135229 	 \)   &  \(	0.0000002177742833 	\) & \(	0.0000000100802685 			 \)  \\
\(	8	\)  &  \(	-0.0000031377435380 	 \)   &  \(	0.0000004295129270 	\) & \(	-0.0000000360000115 			 \)  \\
\(	9	\)  &  \(	-0.0000022459349254 	 \)   &  \(	0.0000004469913959 	\) & \(	-0.0000000647989030 			 \)  \\
\(	10	\)  &  \(	-0.0000012247484604 	 \)   &  \(	0.0000003229989112 	\) & \(	-0.0000000647743390 			 \)  \\
\(	11	\)  &  \(	-0.0000005277189063 	 \)   &  \(	0.0000001662032187 	\) & \(	-0.0000000454779202 			 \)  \\
\(	12	\)  &  \(	-0.0000001871835099 	 \)   &  \(	0.0000000742134452 	\) & \(	-0.0000000243260099 			 \)  \\
\(	13	\)  &\(	-0.0000000544287724 	 \)   &  \(	0.0000000250134666 	\) & \(	-0.0000000097375780 			 \)  \\\hline
\end{tabular}%
\end{table}

\begin{table}[H]
\caption{The Chebyshev coefficients for the middle region with \(u\in[0.8,0.95]\) and \(\delta\in[0.1,0.2]\) to achieve \(O(10^{-8})\) approximation accuracy.
}

\centering
\label{coefficient_region2}
\begin{tabular}{|c| r r r r|}
\hline
		 & \multicolumn{1}{c}{\(0\)}		   &  \multicolumn{1}{c}{\(	1	\)} & \multicolumn{1}{c}{\(2\)} & \multicolumn{1}{c|}{\(	3	 \)}  \\ \hline
\(	0	\) & \(	0.3875945631074440 	\) & \(	0.0026135656106857 	\) & \(	0.0000129352935356 	\) & \(	-0.0000000381872506 	\) \\
\(	1	\) & \(	0.4956109866348150 	\) & \(	0.0008953566937339 	\) & \(	-0.0000067054601666 	\) & \(	-0.0000000125588670 	\) \\
\(	2	\) & \(	0.1197964723594020 	\) & \(	-0.0027935601066155 	\) & \(	-0.0000166388049599 	\) & \(	0.0000000730650209 	\) \\
\(	3	\) & \(	-0.0010763253182188 	\) & \(	-0.0009194311880306 	\) & \(	0.0000085829238628 	\) & \(	0.0000001335013439 	\) \\
\(	4	\) & \(	-0.0023865415596642 	\) & \(	0.0001990751020187 	\) & \(	0.0000035794963500 	\) & \(	-0.0000000521728813 	\) \\
\(	5	\) & \(	0.0004865330909205 	\) & \(	0.0000201874711514 	\) & \(	-0.0000020552522115 	\) & \(	-0.0000000169016304 	\) \\
\(	6	\) & \(	-0.0000106448822976 	\) & \(	-0.0000192517056677 	\) & \(	0.0000002053193871 	\) & \(	0.0000000182461114 	\) \\
\(	7	\) & \(	-0.0000205224390034 	\) & \(	0.0000042306504504 	\) & \(	0.0000001643278424 	\) & \(	-0.0000000050681230 	\) \\
\(	8	\) & \(	0.0000062842657989 	\) & \(	0.0000000655856784 	\) & \(	-0.0000000840329782 	\) & \(	-0.0000000005964520 	\) \\
\(	9	\) & \(	-0.0000007508885322 	\) & \(	-0.0000003315298971 	\) & \(	0.0000000156952882 	\) & \(	0.0000000010012989 	\) \\
\(	10	\) & \(	-0.0000001157724046 	\) & \(	0.0000001086938654 	\) & \(	0.0000000020326053 	\) & \(	-0.0000000003785164 	\) \\
\(	11	\) & \(	0.0000000777284278 	\) & \(	-0.0000000140930457 	\) & \(	-0.0000000022683390 	\) & \(	0.0000000000505659 	\) \\ \hline

\end{tabular}
\end{table}

\begin{table}[H]
\caption{The Chebyshev coefficients for the tail region with \(u\in[0.95,1]\) and \(\delta\in[0.1,0.2]\) to achieve \(O(10^{-8})\) approximation accuracy.
}

\centering
\label{coefficient_region3}
\begin{tabular}{|c| r r r r|}
\hline
		 & \multicolumn{1}{c}{\(0\)}		   &  \multicolumn{1}{c}{\(	1	\)} & \multicolumn{1}{c}{\(2\)} & \multicolumn{1}{c|}{\(	3	 \)}  \\ \hline

\(	0	\) &\(	6.8753214515317200 	\) & \(	0.0154069212636548 	\) & \(	-0.0000209599465735 	\) & \(	-0.0000000864515313 	\) \\
\(	1	\) &\(	8.5433821730433800 	\) & \(	0.0090024260671873 	\) & \(	-0.0000377801621097 	\) & \(	-0.0000000303548005 	\) \\
\(	2	\) &\(	3.4476658871478400 	\) & \(	-0.0130690823892008 	\) & \(	0.0000022840388153 	\) & \(	0.0000000362919660 	\) \\
\(	3	\) &\(	0.9330390647038300 	\) & \(	-0.0087472864187297 	\) & \(	0.0000347919528143 	\) & \(	0.0000000562970750 	\) \\
\(	4	\) &\(	0.1741822695325730 	\) & \(	-0.0023271850423403 	\) & \(	0.0000189337983037 	\) & \(	-0.0000000495179004 	\) \\
\(	5	\) &\(	0.0231441815282366 	\) & \(	-0.0002445078701620 	\) & \(	0.0000029857164456 	\) & \(	-0.0000000266197129 	\) \\
\(	6	\) &\(	0.0027835177479455 	\) & \(	-0.0000075113568411 	\) & \(	-0.0000003271927323 	\) & \(	0.0000000028475004 	\) \\
\(	7	\) &\(	0.0004380896611385 	\) & \(	-0.0000109856591271 	\) & \(	0.0000000145476511 	\) & \(	0.0000000029205683 	\) \\
\(	8	\) &\(	0.0000474342757583 	\) & \(	-0.0000033230341686 	\) & \(	0.0000000850337842 	\) & \(	-0.0000000002011847 	\) \\
\(	9	\) &\(	-0.0000039453038556 	\) & \(	0.0000003821206797 	\) & \(	-0.0000000061400030 	\) & \(	-0.0000000040928856 	\) \\
\(	10	\) &\(	-0.0000006163799253 	\) & \(	0.0000001912780830 	\) & \(	-0.0000000172232771 	\) & \(	0.0000000086484494 	\) \\
\(	11	\) &\(	0.0000004418045229 	\) & \(	-0.0000000406903687 	\) & \(	-0.0000000104246286 	\) & \(	0.0000000086001421 	\) \\
\(	12	\) &\(	0.0000000377188891 	\) & \(	-0.0000000175539467 	\) & \(	-0.0000000096591111 	\) & \(	0.0000000085524035 	\) \\
\(	13	\) &\(	-0.0000000336347029 	\) & \(	0.0000000003271218 	\) & \(	-0.0000000102016363 	\) & \(	0.0000000085537387 	\) \\
\hline
   \end{tabular}
\end{table}
\clearpage
\bibliographystyle{abbrv}

\bibliography{ref2}
\end{document}